\documentclass[11pt]{article}

\usepackage[letterpaper, margin=1.1in] {geometry}
\usepackage{indentfirst}
\usepackage{amsfonts}
\usepackage{amsmath}
\usepackage{amsthm}
\usepackage{color}
\usepackage{eufrak}
\usepackage{verbatim}
\usepackage[makeroom]{cancel}

\usepackage{amssymb}
\usepackage{amsthm}
\usepackage{mathrsfs}
\usepackage{epsfig}
\usepackage{makeidx}
\usepackage{graphicx}
\usepackage{indentfirst}
\usepackage[numbers]{natbib}
\usepackage{color}
\usepackage[backref=page]{hyperref}
\usepackage{graphicx}
\usepackage{graphicx}
\usepackage{hyperref}
\usepackage{graphicx}

\definecolor{corlinks}{RGB}{200,0,40}
\definecolor{cormenu}{RGB}{200,0,40}
\definecolor{corurl}{RGB}{200,0,40}

\hypersetup{
colorlinks=true,
urlcolor=corlinks,
linkcolor=corlinks,
menucolor=cormenu,
citecolor=corlinks,
pdfborder= 0 0 0
}

\newtheorem{theorem}{Theorem}
\newtheorem{lemma}{Lemma}
\newtheorem{corollary}{Corollary}
\newtheorem{definition}{Definition}
\newtheorem{proposition}{Proposition}
\newtheorem{remark}{Remark}

\newcommand{\eqdef}{\stackrel{\rm def}{=}}

\DeclareMathOperator{\poly}{poly}

\DeclareMathOperator{\io}{i.o.}

\def\colorful{1}

\ifnum\colorful=1

\fi
\ifnum\colorful=0

\fi

\newcommand{\esymb}{\diamond}

\begin{document}

\title{Pseudodeterministic Constructions in Subexponential Time\vspace{0.4cm}}

\author{Igor C. Oliveira\footnote{ This work received support from CNPq grant 200252/2015-1.}\\
  \small{Charles University in Prague}\\
  \and Rahul Santhanam\footnote{\texttt{rahul.santhanam@cs.ox.ac.uk.} Supported by the European Research Council under the European Union's Seventh Framework Programme (FP7/2007-2014)/ERC Grant Agrement no. 615075.}\\
  \small{University of Oxford}\\~\\
}


\maketitle


\begin{abstract}
We study {\it pseudodeterministic constructions}, i.e., randomized algorithms which output the {\it same solution} on most computation paths. We establish unconditionally that there is an infinite sequence $\{p_n\}_{n \in \mathbb{N}}$ of increasing primes and a randomized algorithm $A$ running in expected sub-exponential time such that for each $n$, on input $1^{|p_n|}$, $A$ outputs $p_n$ with probability $1$. In other words, our result provides a pseudodeterministic construction of primes in sub-exponential time which works infinitely often.

\vspace{0.05cm}
This result follows from a much more general theorem about pseudodeterministic constructions. A property $Q \subseteq \{0,1\}^{*}$ is $\gamma$-dense if for large enough $n$, $|Q \cap \{0,1\}^n| \geq \gamma 2^n$. We show that for each $c > 0$ at least one of the following holds: (1) There is a pseudodeterministic polynomial time construction of a family $\{H_n\}$ of sets,  $H_n \subseteq \{0,1\}^n$, such that for each $(1/n^c)$-dense property $Q \in \mathsf{DTIME}(n^c)$ and every large enough $n$, $H_n \cap Q \neq \emptyset$; or (2) There is a deterministic sub-exponential time construction of a family $\{H'_n\}$ of sets, $H'_n \subseteq \{0,1\}^n$, such that for each $(1/n^c)$-dense property $Q \in \mathsf{DTIME}(n^c)$ and for infinitely many values of $n$, $H'_n \cap Q \neq \emptyset$.

\vspace{0.05cm}

We provide further algorithmic applications that might be of independent interest. Perhaps intriguingly, while our main results are unconditional, they have a non-constructive element, arising from a sequence of applications of the hardness versus randomness paradigm. 
\end{abstract}

\newpage

\section{Introduction}\label{s:introduction}

Number theory tells us that a significant fraction of $n$-bit integers are prime, but can we efficiently generate an $n$-bit prime deterministically? This is a fundamental question in computational number theory and cryptography. A naive solution is to test the $n$-bit numbers for primality in order, starting with $2^{n-1} + 1$, using the AKS primality test \citep{Agrawal02primesis}, until we find one that is prime. Known results about the distribution of prime numbers guarantee that this procedure works in time $2^{0.525n + o(n)}$ \citep{Baker-Harman-Pintz01}. Despite all the progress that has been made in understanding the behaviour of prime numbers, the best known deterministic algorithm for generating $n$-bit primes runs in time $2^{n/2+o(n)}$ \citep{DBLP:journals/jal/LagariasO87}, which is not significantly better.

If we are allowed randomness in our generating algorithm, the problem becomes almost trivial: we repeatedly guess an $n$-bit number at random and test it for primality, halting if a prime is found. Using the Prime Number Theorem, we will succeed with probability $1-o(1)$ after $\widetilde{O}(n)$ tests, each of which can be implemented in $\poly(n)$ time. Thus the efficient deterministic generation problem reduces to {\it derandomizing} this algorithm. Under a strong hardness hypothesis, such as the assumption that linear exponential time requires exponential size circuits almost everywhere, this derandomization can be performed using known results from the theory of hardness-randomness tradeoffs \citep{DBLP:conf/stoc/ImpagliazzoW97}. However, we appear to be very far from proving such   
strong circuit lower bounds.

A few years ago, the Polymath 4 project considered precisely this question of efficient generation of primes, with the goal of using the state of the art in complexity theory and in number theory to obtain better results. It was observed during the project that several famous number-theoretic conjectures, such as Cramer's conjecture and Schinzel's hypothesis H, imply better generation algorithms, as do circuit lower bound assumptions, as described above. All of these conjectures seem far beyond our reach at present. The Polymath 4 project established \citep{MR2869058} an improved algorithm for determining the {\it parity} of the number of primes in a large interval, but this has not yet yielded an unconditional improvement to the best known deterministic generation algorithm. The project posed the following question of interest to complexity theorists: does $\mathsf{BPP} = \mathsf{P}$ imply more efficient deterministic generation of primes? This is not clear because the randomized generation algorithm for primes is not an algorithm for a decision problem.

Given the difficulty of finding more efficient deterministic generation algorithms, a natural approach is to relax the question. We know that randomized generation is easy, and we do not know of good deterministic generation algorithms. Is there an intermediate notion that could perhaps be useful?

Such a notion of {\it pseudodeterministic algorithms} was defined by Goldwasser and Gat \citep{DBLP:journals/eccc/GatG11}, motivated by applications in cryptography and distributed computing, and further studied in \citep{DBLP:conf/innovations/GoldreichGR13, DBLP:journals/eccc/GoldwasserG15, DBLP:journals/eccc/Grossman15}. A pseudodeterministic algorithm is a randomized algorithm, but the algorithm generates the {\it same} output with high probability. Thus, while the algorithm is not deterministic, the output of the algorithm {\it looks} deterministic to a computationally bounded observer. In the context of generating primes, pseudodeterministic generation means that the {\it same} prime is output on most computation paths. Note that the naive random generation algorithm does not have this property.

Goldwasser and Gat  \citep{DBLP:journals/eccc/GatG11} describe the question of pseudodeterministic generation of primes as ``perhaps the most compelling challenge for finding a unique output''. To the best of our knowledge, it is not known how to do better for this problem than to use the fastest deterministic algorithm.

\subsection{Main Results}\label{ss:our_results}

The main application of our techniques is that there is unconditionally a subexponential-time pseudodeterministic algorithm for generating infinitely many primes. 

\begin{theorem}
\label{t:primes}
Let $\varepsilon > 0$ be any constant. There is an infinite sequence $\{p_n\}_{n \in \mathbb{N}}$ of increasing primes, and a randomized algorithm $A$ running in expected time $O(2^{m^{\varepsilon}})$ on inputs of length $m$ such that for each $n \in \mathbb{N}$, $A$ on input $1^{|p_n|}$ outputs $p_n$ with probability $1$.
\end{theorem}

Note that the algorithm $A$ is \emph{zero-error}: on input $1^{|p_n|}$, it outputs the same prime $p_n$ on all computation paths. In fact, on any input $1^m$, the algorithm either operates deterministically and halts without output, or else it outputs the same $m$-bit prime on all computation paths.

Using the proof idea of Theorem \ref{t:primes}, we answer a variant of the question from Polymath 4: under the assumption that $\mathsf{ZPP} \subseteq \io \mathsf{DTIME}(2^{n^c})$ for some fixed $c > 0$ (which is much weaker than $\mathsf{BPP} = \mathsf{P}$), we show that there is a subexponential-time \emph{deterministic} algorithm for generating infinitely many primes (see Theorem \ref{t:cond_detprimes} in Section \ref{s:uncond_cons}).

Theorem \ref{t:primes} has some unusual features: we show that an algorithm {\it exists} satisfying the described properties, but we are not able to {\it explicitly} give such an algorithm. Similarly, while we show the existence of an infinite sequence of primes $\{p_n\}$, we are unable to bound $|p_{n+1}|$ as a function of $|p_n|$. These non-constructive features, which are surprising in the context of an algorithmic solution to a natural problem, arise because Theorem \ref{t:primes} is proved as a corollary of the following generic result about explicit constructions.

\begin{theorem}
\label{t:generic}
Call a property $Q \subseteq \{0,1\}^{*}$ $\gamma$-dense if for each large enough $n$, $|Q \,\cap\, \{0,1\}^n|/2^n \geq \gamma(n)$. For each $c > 0$ at least one of the following holds\emph{:}

\begin{enumerate}
\item[\emph{1.}] There is a deterministic sub-exponential time construction of a family $\{H_n\}$ of sets, $H_n \subseteq \{0,1\}^n$, such that for each $(1/n^c)$-dense property $Q \in \mathsf{DTIME}(n^c)$ and for infinitely many values of $n$, $H_n \cap Q \neq \emptyset$.

\item[\emph{2.}] There is a zero-error pseudodeterministic polynomial time algorithm outputting a family $\{H'_n\}$ of sets,  $H'_n \subseteq \{0,1\}^n$, such that for each $(1/n^c)$-dense property $Q \in \mathsf{DTIME}(n^c)$ and every large enough $n$, $H'_n \cap Q \neq \emptyset$.
\end{enumerate}

\end{theorem}

We derive Theorem \ref{t:primes} by taking $Q$ to be the set of primes in Theorem \ref{t:generic}, and observing that the statement of Theorem \ref{t:primes} follows both from the first and second item of Theorem \ref{t:generic}. The non-constructivity comes from not knowing which of these two items holds, and we discuss this issue in more detail in Section \ref{s:discussion}.

Consider any property $Q$ as in the statement of Theorem \ref{t:generic}, namely with polynomial density and decidable in deterministic polynomial time. For any such property, there is a randomized generation algorithm in polynomial time analogous to the one we described for Primes: generate strings of length $n$ at random and test for membership in $Q$, halting if a string in $Q$ is found. What Theorem \ref{t:generic} says is that this randomized generation algorithm can {\it unconditionally} be made pseudodeterministic in a generic fashion, albeit at the cost of increasing the running time to subexponential, and only succeeding for infinitely many input lengths $n$.

Theorem \ref{t:generic} is a very general result, and the generality of our techniques enables us to apply them to give unconditional pseudodeterministic algorithms in other contexts. Even if Theorem \ref{t:primes} does not give an explicit algorithm for generating primes of every input length, it does provide evidence that such algorithms might not be out of reach of current techniques. More explicit and more efficient algorithms could perhaps be designed by tailoring them more closely to the property at hand.

Nevertheless, we are interested in the question of whether there is a fundamental bottleneck to constructivity in our proof technique for Theorem \ref{t:generic}. By refining some aspects of our approach, we are able to get {\it somewhat explicit} algorithms for generating primes as follows.

\begin{theorem}
\label{t:explicit}
For each $\varepsilon > 0$, there is a constant $k > 1$, an infinite sequence $\{q_n\}_{n \in \mathbb{N}}$ of increasing primes, and an explicit randomized algorithm $A$ halting in time $O(2^{m^{\varepsilon}})$ on inputs of length $m$ such that for each $n \in \mathbb{N}$, $A(1^{|q_n|})$ outputs $q_n$ with probability $1-o(1)$, and moreover $|q_{n+1}| < |q_n|^k$ for each $n \in \mathbb{N}$. 
\end{theorem}

The algorithm in Theorem \ref{t:explicit} is explicit, but we still cannot say for sure on which input lengths it will succeed. We do have the guarantee that the gaps between successive input lengths on which the algorithm succeeds are not superpolynomially large. Theorem \ref{t:explicit} does not strictly improve on Theorem \ref{t:primes} -- it is not necessarily true, for example, that the \emph{same} prime is output on \emph{all} computation paths. Namely, we get a \emph{bounded-error} pseudodeterministic algorithm rather than a zero-error one. However, the issue of non-constructivity is somewhat mitigated in Theorem \ref{t:explicit}.

Theorem \ref{t:generic} yields subexponential-time pseudodeterministic generation algorithms for properties that are both easy and dense. A more general context is polynomial-time {\it sampling}. Here we are given a sampler that either outputs a string (with high probability) or aborts, and we wish to design a new sampler that outputs a \emph{fixed} string with high probability, with the constraint that this fixed string should belong to the range of the original sampler. We show how to adapt our ideas to give an analogue of Theorem \ref{t:generic} for this problem. We refer to Section \ref{ss:samplable} for more details.

We also study the relationship between pseudodeterminism and derandomization. A natural problem associated with derandomization is the Circuit Acceptance Probability Problem ($\mathsf{CAPP}$): given a circuit $C$, output a small-error additive approximation to the acceptance probability of $C$. There is a simple randomized polynomial-time algorithm for this problem -- sample polynomially many inputs of $C$ at random, and output the fraction of these inputs on which $C$ accepts. This algorithm outputs different estimations of the acceptance probability of $C$ on different computation paths. We show how to get a pseudodeterministic algorithm as in Theorem \ref{t:generic} for infinitely many input lengths, but our algorithm is only guaranteed to succeed distributionally rather than in the worst case. (We refer to Section \ref{s:pseudo_derand} for a precise formulation of the result.)

\begin{theorem}
\label{t:capp}
For any $\varepsilon > 0$ and any polynomial-time samplable sequence $\mathfrak{D} = \{\mathcal{D}_n\}$ of distributions over Boolean circuits, $\mathsf{CAPP}$ is pseudodeterministically solvable for infinitely many input lengths $n$ in time $2^{O(n^{\varepsilon})}$ with high probability over $\mathcal{D}_n$.
\end{theorem}
 
One of the main questions in the theory of derandomization concerns the relationship between ``white-box'' and ``black-box'' derandomization. Black-box derandomization refers to derandomization using a fixed pseudorandom generator, while in white-box derandomization, we are given a circuit and asked to approximate its acceptance probability deterministically. Black-box derandomization implies white-box derandomization of randomized algorithms, and white-box derandomization implies deterministic simulations of languages in $\mathsf{BPP}$. However, it is unknown whether these implications can be reversed in general (cf.~\citep{DBLP:conf/coco/Fortnow01, DBLP:books/sp/goldreich2011/Goldreich11g}), and separations are known in relativized worlds \citep{DBLP:conf/coco/Fortnow01}. 

We prove that these notions are all in fact {\it equivalent} in the setting of subexponential-time derandomization that works infinitely often on average over polynomial-time samplable distributions.  (While Theorem \ref{t:capp} provides a randomized algorithm, we stress that the algorithms postulated in Theorem \ref{t:equiv} below are all \emph{deterministic}. We refer to Section \ref{s:pseudo_derand} for definitions.)

\begin{theorem}
\label{t:equiv}
The following statements are equivalent\emph{:}
\begin{enumerate}
\item[\emph{1.}] For each polynomial-time samplable distribution $\mathfrak{D}$ of Boolean circuits and each $\varepsilon > 0$, there
is an \emph{i.o.PRG} $\mathfrak{G}$ on average over $\mathfrak{D}$ with seed length $n^{\varepsilon}$ that is computable in time $2^{O(n^{\varepsilon})}$.

\item[\emph{2.}] For each polynomial-time samplable distribution $\mathfrak{D}$ over Boolean circuits and each $\varepsilon > 0$, $\mathsf{CAPP}$ is solvable infinitely often in time $2^{O(n^{\varepsilon})}$ on average over $\mathfrak{D}$.

\item[\emph{3.}] For each polynomial-time samplable distribution $\mathfrak{D}$ over input strings and each $\varepsilon > 0$, $\mathsf{BPP}$ is solvable infinitely often in time $2^{O(n^{\varepsilon})}$ with $O(\log(n))$ bits of advice on average over $\mathfrak{D}$.

\item[\emph{4.}] For each $\varepsilon > 0$, $\mathsf{BPP}$ is solvable infinitely often in time $2^{O(n^{\varepsilon})}$ on average over $U_n$.
\end{enumerate}
\end{theorem} 

\vspace{0.2cm}

Therefore, in order to establish all these items it is necessary and sufficient to turn the (unconditional) pseudodeterministic algorithm from Theorem \ref{t:capp} into a deterministic algorithm.

\subsection{Related Work}\label{ss:related_work}

There has been a lot of work on explicitly constructing combinatorial objects that can be shown to exist by the probabilistic method. Vadhan \citep{TCS-010} surveys known unconditional constructions of such objects, and connections to the theory of pseudorandomness. There are important properties, such as the Ramsey property, for which optimal explicit constructions are still not known, though there has been much recent progress (cf.~\citep{DBLP:conf/stoc/ChattopadhyayZ16, DBLP:conf/stoc/Cohen16}). However, in many cases, such constructions do exist if the notion of explicitness is weakened or if a sufficiently strong derandomization hypothesis is made \citep{MR2275721, DBLP:journals/mst/Santhanam12}. 

The techniques used to show explicit constructions of combinatorial objects do not seem to be directly relevant to constructions of algebraic objects such as irreducible polynomials or number-theoretic objects such as primes. Shoup \citep{Shoup90} shows how to deterministically construct an irreducible polynomial of degree $n$ over the field $\mathbb{F}_p$ in time polynomial in $n$ and $p$. This is useful for constructions of irreducible polynomials over fields of small characteristic, but the large-characteristic case is still open -- ideally we would like the construction algorithm to operate in time $\poly(n, \log(p))$. 

For primes, known results are even weaker. The fastest known algorithm generates $n$-bit primes in time $2^{(\gamma + o(1)) n}$, for $\gamma = 1/2$ \citep{DBLP:journals/jal/LagariasO87}. There are algorithms that achieve an arbitrarily small constant $\gamma > 0$, but assume certain number-theoretic conjectures. Faster running times can be obtained under stronger conjectures, and we refer to the introduction of \citep{MR2869058} for more information. 

Assuming standard (but hard to prove) derandomization hypotheses, stronger explicit constructions are known. 
It is folklore that the existence of hitting sets against a class $\mathfrak{C}$ of algorithms  yields explicit deterministic constructions for every dense property computed by $\mathfrak{C}$-algorithms. References that discuss this include \citep{MR2275721, DBLP:journals/mst/Santhanam12, DBLP:books/sp/goldreich2011/Goldreich11g}). As an example, the assumption that $\mathsf{E}$ requires exponential-size Boolean circuits implies polynomial-time explicit constructions of irreducible polynomials and of primes. However, in this work, we are interested in {\it unconditional}, albeit pseudodeterministic, constructions for properties decidable in polynomial time. (While there are unconditional pseudorandom generators against restricted classes of algorithms such as small-depth polynomial size circuits, it is known that such circuits cannot decide several properties of interest, such as prime numbers \citep{DBLP:journals/jcss/AllenderSS01}.)

Perhaps the most closely related result is a certain unconditional derandomization of one-sided polynomial time randomized algorithms in subexponential time established in \citep{DBLP:journals/jcss/Kabanets01}. However, the focus there is on decision problems, and as remarked above, such results do not immediately imply pseudodeterministic constructions. The {\it easy witness} method from \citep{DBLP:journals/jcss/Kabanets01} is an important tool in some of our proofs.

The investigation of the power and limitations of pseudodeterministic algorithms is fairly recent. A sequence of works \citep{DBLP:journals/eccc/GatG11, DBLP:conf/innovations/GoldreichGR13, DBLP:journals/eccc/GoldwasserG15, DBLP:journals/eccc/Grossman15} has started the development of a more general theory of pseudodeterministic computation. These papers give pseudodeterministic algorithms that are faster than their known deterministic counterparts for various concrete search problems of interest, and also give some structural results.

Connections between black-box derandomization, white-box derandomization and deterministic simulation of $\mathsf{BPP}$ are explored by Fortnow \citep{DBLP:conf/coco/Fortnow01}, who shows relativized separations between these notions. Goldreich \citep{DBLP:books/sp/goldreich2011/Goldreich11g} shows an equivalence between white-box derandomization and black-box derandomization in the average-case setting; however he consider pseudo-random generators that work against uniform algorithms, rather than against non-uniform circuits. Implagliazzo, Kabanets and Wigderson \citep{DBLP:journals/jcss/ImpagliazzoKW02} give an equivalence between white-box and black-box derandomization in the setting where we allow the derandomization to be done {\it non-deterministically} with a small amount of advice, and the derandomization is only required to be correct for infinitely many input lengths.

\subsection{Techniques}\label{ss:techniques}

Suppose we wish to pseudodeterministically generate strings satisfying some property $Q$ that is dense and is decidable in polynomial time, such as $\mathsf{Primes}$.\footnote{For concreteness, we let $\mathtt{Primes} \eqdef \{w \in \{0,1\}^* \mid w = w_{n-1} w_{n-2} \ldots w_0,~w_{n-1} = 1,~\text{and}~\sum_{i = 0}^{n-1}2^i w_i~\text{is prime}\,\}$.} It is well-known that using standard hardness-randomness tradeoffs, efficient deterministic generation can be done under a strong enough circuit lower bound assumption for $\mathsf{E} \eqdef \mathsf{DTIME}(2^{O(n)})$ (linear exponential time).

Our first observation is that efficient pseudodeterministic generation follows from {\it weaker} circuit lower bound assumptions, namely circuit lower bounds for the probabilistic classes $\mathsf{BPE} = \mathsf{BPTIME}(2^{O(n)})$ (bounded-error linear exponential time) and $\mathsf{ZPE} = \mathsf{ZPTIME}(2^{O(n)})$ (zero-error linear exponential time). Pseudodeterministic algorithms come in two flavours: \emph{zero-error}, where the algorithm outputs a fixed string (with high probability) or else aborts, and \emph{bounded-error}, where it outputs a fixed string on most computation paths but might output other strings on exceptional computation paths. It turns out that strong enough circuit lower bounds for $\mathsf{BPE}$ imply efficient bounded-error pseudodeterministic generation, and strong enough circuit lower bounds for $\mathsf{ZPE}$ imply efficient zero-error pseudodeterministic generation. This is because a hard function in $\mathsf{BPE}$ (resp.~a hard function in $\mathsf{ZPE}$) yields a bounded-error (resp.~zero-error) pseudodeterministic construction of a \emph{discrepancy set} (see Section \ref{s:preliminaries}). In turn, once such a discrepancy set is obtained, its elements can be checked in some fixed order for membership in $Q$, and the first such element belonging to $Q$ can then be output pseudodeterministically.

Note that circuit lower bounds for $\mathsf{BPE}$ do not seem useful for decision problems: using a sufficiently hard function in $\mathsf{BPE}$ in a pseudorandom generator would merely yield the tautologous inclusion $\mathsf{BPP} \subseteq \mathsf{BPP}$. However, such circuit lower bounds {\it are} useful for pseudo-derandomization of search problems.\footnote{It is natural to wonder if there are techniques tailored to $\mathsf{BPE}$ and $\mathsf{ZPE}$ lower bounds. We refer the interested reader to the approaches outlined in \citep{1611.01190}.}

To turn our conditional pseudodeterministic generation algorithms into unconditional ones, we use the uniform hardness-randomness tradeoffs of \citep{DBLP:journals/jcss/ImpagliazzoW01, DBLP:journals/cc/TrevisanV07} and win-win analysis, showing that both the success and the failure of certain pseudorandom generators can be exploited to give non-trivial pseudodeterministic constructions. The specific win-win analysis we use depends on the intended application. In order to make some constructions zero-error, we further combine the arguments with the easy witness method \citep{DBLP:journals/jcss/Kabanets01}.

We describe next the main conceptual ideas behind the proof of each main result stated in Section \ref{ss:our_results}.\\
\\
{\noindent \bf Theorem \ref{t:primes}} (Non-constructive zero-error pseudo-deterministic algorithm for $\mathtt{Primes}$).~Recall that Theorem \ref{t:primes} is a straightforward consequence of Theorem \ref{t:generic}. In each case, it follows that there is a pseudodeterministic subexponential time algorithm which simply outputs the lexicographic first $n$-bit prime in the hitting set. The non-constructive aspect of Theorem \ref{t:primes} arises because we do not know which of items (1) and (2) holds in Theorem \ref{t:generic}.\\
\\
{\noindent \bf Theorem \ref{t:generic}} (Unconditional pseudo-deterministic hitting set generators).~Here we use a win-win-win analysis. We try two different candidate hitting set generators $\mathcal{H}^{\mathsf{easy}}$ and $\mathcal{H}^{\mathsf{hard}}$, both computable in deterministic subexponential time. The first generator is based on the easy witness method of Kabanets \cite{DBLP:journals/jcss/Kabanets01}, and the second is the generator of Trevisan-Vadhan \cite{DBLP:journals/cc/TrevisanV07}. If either of the corresponding hitting set families intersects $Q \cap \{0,1\}^n$ infinitely often, we have that the first item of Theorem \ref{t:generic} holds. If not, then we show that we have the complexity collapse $\mathsf{PSPACE} = \mathsf{ZPP}$. We exploit this collapse to derive  strong circuit lower bounds in $\mathsf{ZPE}$, and are then able to use our conditional argument for efficient pseudodeterministic generation to conclude the second item.\\
\\
{\noindent \bf Theorem \ref{t:explicit}} (Constructive bounded-error pseudo-deterministic algorithm for $\mathtt{Primes}$).~The proof for Theorem \ref{t:generic} gives no information on the sparsity of input lengths for which the construction succeeds, nor does it give us an explicit algorithm. An important reason for this is that we prove item (2) to hold in Theorem \ref{t:generic} only if the candidate hitting set generators in item (1) fail almost everywhere. We first show a refinement of the main technical lemma which allows us to derive a consequence for item (2) from a failure of item (1) for all input lengths in a polynomially large range of input lengths. We then crucially use the fact that item (1) gives deterministically generatable hitting sets which can be checked for correctness somewhat efficiently to guarantee that either item (1) or item (2) must hold in every large enough polynomial range of input lengths.\footnote{The possible sparsity of the sequence of primes produced by our algorithm is an artifact of our complexity-theoretic techniques. In contrast, constructions using standard number-theoretic techniques typically use analytic methods and succeed for every large enough input length (James Maynard, personal communication).} The details are somewhat technical, and we provide a more intuitive explanation in Section \ref{s:explicit}.\\
\\
{\noindent \bf Theorem \ref{t:capp}} (Pseudodeterministic algorithm for $\mathsf{CAPP}$).~We use a win-win analysis once again, relying on the fact that pseudo-random generators can be employed to approximate the acceptance probability of circuits in a deterministic fashion. If the candidate pseudo-random generator we use succeeds infinitely often, we can solve $\mathsf{CAPP}$ deterministically infinitely often; if it fails in an average-case sense, we adapt the arguments used in the proof of Theorem \ref{t:generic} to show how to solve $\mathsf{CAPP}$ pseudodeterministically in polynomial time.\\
\\
{\noindent \bf Theorem \ref{t:equiv}} (Equivalences for deterministic algorithms in the average-case ).~The most interesting implications rely on results from previous works. The simulation in (3) assuming (2) introduces advice, and this advice is eliminated in item (4) as follows.   Under assumption (3), it follows by diagonalization that $\mathsf{EXP} \neq \mathsf{BPP}$. In turn, \citep{DBLP:journals/jcss/ImpagliazzoW01} showed that this hypothesis provides sub-exponential time simulations of $\mathsf{BPP}$, as required in (4). Perhaps a more surprising connection is that (4) implies (1). Our argument is based on another win-win analysis, and relies on  results from \citep{DBLP:journals/cc/BabaiFNW93, DBLP:journals/jcss/ImpagliazzoW01, DBLP:journals/cc/TrevisanV07}. Under (4), it can be shown that $\mathsf{EXP} \neq \mathsf{BPP}$, and that this implies that either $\mathsf{EXP} \nsubseteq \mathsf{P}/\mathsf{poly}$, or $\mathsf{PSPACE} \neq \mathsf{BPP}$. In each case, we argue that any samplable distribution admits an average-case pseudo-random generator satisfying the conditions in (1).

\subsection{Constructivity: PSEUDO versus SPARSE}\label{s:discussion}

Note the quantification in the statement of Theorem \ref{t:generic}: there is a pseudodeterministic polynomial time algorithm producing a hitting set family which works for {\it every} sufficiently dense easy property, or there is a deterministic subexponential time algorithm producing hitting set families which work infinitely often for {\it every} sufficiently dense easy property. Thus we provably live in one of two worlds:\footnote{Inspired by the worlds of Impagliazzo \citep{Impagliazzo95}.} a world {\bf PSEUDO} where there is a generic hitting set family for sufficiently dense easy properties that is computable in zero-error pseudodeterministic polynomial time, or a world {\bf SPARSE} where generic hitting set families are computable deterministically but only in subexponential time, and moreover the hitting sets only work infinitely often. These two worlds could co-exist. Indeed many complexity theorists believe that linear exponential time requires exponential size circuits on almost all input lengths, which would imply the co-existence of the two worlds. However, given the available techniques at this point in time, we are only able to establish that we live in one of the two worlds.

The fact that we can show that we live in one of these two worlds, but don't know which one, leads to a certain {\it non-constructivity} in our results, which is ironic given that we are motivated by explicit constructions! We can show that we live in one of the worlds {\bf PSEUDO} or {\bf SPARSE}, and in the latter case, our proof provides an explicit algorithm witnessing that we live in that world, namely the algorithm producing the union of the hitting set families $\mathcal{H}^{\mathsf{easy}}$ and $\mathcal{H}^{\mathsf{hard}}$. However, in the former case, the proof does not provide an explicit algorithm witnessing that we live in that world. This is because the proof relies on the existence of an easy dense property which acts as distinguisher for the hitting set families, and we do not know a priori which property this is. 

A further element of non-constructivity in Theorem \ref{t:generic} is that in case we live in the world {\bf SPARSE}, we have no information on the set of input lengths for which the hitting set works, except that this set is infinite. This set could be arbitrarily sparse. The reason is that we can only show the second item based on the hitting set families $\mathcal{H}^{\mathsf{easy}}$ and $\mathcal{H}^{\mathsf{hard}}$ failing to work on all large enough input lengths.

We partially address both issues in Section \ref{s:explicit}, and we refer to that section for further discussions on {\bf PSEUDO} and {\bf SPARSE}.

\section{Preliminaries}\label{s:preliminaries}

\subsection{Basic Definitions and Notation}\label{ss:basic_def_notation}

We refer to algorithms in a sufficiently strong computational model, such as the 2-tape Turing Machine over the alphabet $\{0,1, \perp\}$. We use standard encodings of objects such as numbers and finite sets by strings. For convenience, we let $\mathbb{N} \eqdef \{1,2, \ldots\}$. Given a string $u \in \{0,1\}^{\geq n}$, we let $\mathsf{left}_{n}(u)$ denote its leftmost $n$ bits.\\

\noindent \textbf{Properties and Density.} A {\it property}\footnote{We use the term ``property'' rather than ``language'' or ``Boolean function'' to highlight that in our setting, we are more interested in the {\it efficient constructibility} of positive instances than in efficient recognizability.} is an arbitrary set $Q \subseteq \{0,1\}^*$. Let $Q_n \eqdef Q \,\cap\, \{0,1\}^n$, and $Q_{\leq n} \eqdef \bigcup_{m\leq n} Q_m$. Given a function $\gamma \colon \mathbb{N} \to \mathbb{R}$, we say that $Q$ is $\gamma$-dense if $|Q_n|/2^n \geq \gamma(n)$ for every large enough $n$. We say that $Q$ is {\it dense} if there is $c > 0$ such that $Q$ is $\gamma$-dense with $\gamma(n) \geq n^{-c}$ for every sufficiently large $n$. We say that $Q$ is {\it easy} if $Q$ is decidable in deterministic polynomial time. \\
\\
\noindent \textbf{Hitting Sets and Discrepancy Sets.} Let $\mathcal{H} = \{H_n\}_{n \in \mathbb{N}}$ be a sequence, where $H_n \subseteq \{0,1\}^n$ for each $n \in \mathbb{N}$, and let $Q \subseteq \{0,1\}^*$ be a property. We say that $\mathcal{H}$ is a {\it hitting set} for $Q$ at length $n$ if $H_n \cap Q_n \neq \emptyset$. Given a set $S \subseteq \mathbb{N}$ of input lengths, we say that $\mathcal{H}$ is a {\it hitting set family} for $Q$ on $S$ if it is a hitting set for $Q$ at length $n$ for each $n \in S$. We say that $\mathcal{H}$ is a hitting set family for $Q$ in case it is a hitting set family for $Q$ on some co-finite set $S$, and we say that $\mathcal{H}$ is an {\it i.o.~hitting set family} for $Q$ if it is a hitting set for $Q$ on some infinite set $S$. 

Let $\gamma \colon \mathbb{N} \to \mathbb{R}$ be an arbitrary function. We say that a sequence $\mathcal{H}$ of (multi-) sets $H_n$ is a {\it $\gamma$-discrepancy set} for $Q$ at length $n$ if $\left| |Q_n|/2^n - |H_n \cap Q_n|/|H_n| \right| < \gamma(n)$. The notions of being a $\gamma$-discrepancy set family at a set $S \subseteq \mathbb{N}$ of input lengths, an {\it i.o.~$\gamma$-discrepancy set family} and a $\gamma$-discrepancy set family are defined in analogy to the previous paragraph. Note that if $Q$ is $\gamma$-dense, then a $\gamma$-discrepancy set family $\mathcal{H}$ for $Q$ is also a hitting set family for $Q$. \\
\\
\noindent \textbf{Pseudodeterministic Algorithms.} By default, we consider randomized algorithms whose input and output are both strings. Let $p$ be a real-valued parameter. A randomized algorithm $A(\cdot)$ is said to be a {\it bounded-error pseudodeterministic} algorithm with success probability $p$ if there is a function $f\colon \{0,1\}^{*} \to \{0,1\}^{*}$ such that $A(x) = f(x)$ with probability at least $p$ over the random choices of $A$. We call such a function $f$ a {\it $p$-canonical function} for $A$, and an output $f(x)$ of the function a {\it $p$-canonical output} of $A(x)$. A randomized algorithm $A(\cdot)$ is said to be a {\it zero-error pseudodeterministic} algorithm with success probability $p$ if there is a function $f\colon \{0,1\}^{*} \to \{0,1\}^{*}$ such that on each computation path, $A(x)$ outputs either $f(x)$ or $\perp$, and $A(x) = f(x)$ with probability at least $p$ over the random choices of $A$. 

Given a set $S$ of inputs, we say that an algorithm is {\it pseudodeterministic on $S$} if the pseudodeterministic promise is only required to hold for inputs in $S$. When the algorithm takes a unary input, we sometimes identify subsets of inputs with subsets of $\mathbb{N}$, in the natural way.

For any bounded-error pseudodeterministic algorithm with success probability $> 1/2$ or any zero-error pseudodeterministic algorithm with success probability $> 0$, there is a {\it unique} $p$-canonical function, defined in the former case by the unique output produced with probability $> 1/2$, and in the latter case by the unique non-$\perp$ output. In both cases, we simply call the function and the corresponding output {\it canonical}.

Let $Q$ be a property, $p$ a real-valued parameter, and $T\colon\mathbb{N} \rightarrow \mathbb{N}$ be a time function. We say that there is a {\it bounded-error pseudodeterministic construction} with success probability $p$ for $Q$ in time $T$ if there is a bounded-error pseudodeterministic algorithm $A$ with success probability $p$  such that for all large enough $n$, $A(1^n)$ halts in time $T(n)$ and has a $p$-canonical output in $Q_n$. Similarly, we say that there is a {\it zero-error pseudodeterministic construction} with success probability $p$ for $Q$ in time $T$ if there is a zero-error pseudodeterministic algorithm $A$ with success probability $p$, such that for all large enough $n$, $A(1^n)$ halts in time $T(n)$ and has a $p$-canonical output in $Q_n$.

As with hitting sets and discrepancy sets, given a set $S \subseteq \mathbb{N}$, we can generalize the notion of a pseudodeterministic construction to the notion of a pseudodeterministic construction on $S$, which means that $A(1^n)$ has a $p$-canonical output in $Q_n$ for each $n \in S$. An {\it i.o.~pseudodeterministic construction} is a pseudodeterministic construction on some infinite set $S$.

When the parameter $p = p(n)$ for a pseudodeterministic algorithm is omitted, it is taken by default to be $1-o(1)$. For such a pseudodeterministic algorithm $A$ on input $x$ of length $n$, for a large enough $n$, we will sometimes abuse notation and use $A(x)$ to refer to the canonical output of $A$ on $x$. Indeed, as observed next, bounded-error pseudodeterministic constructions and zero-error pseudodeterministic constructions are largely robust to variations in the success probability. 

\begin{proposition}
\label{p:pseudorobust}
Let $Q \subseteq \{0,1\}^*$ be a property, and $T\colon \mathbb{N} \rightarrow \mathbb{N}$ be a time function. 
\begin{itemize}
\item[\emph{(}i\emph{)}] There exists a bounded-error pseudodeterministic construction for $Q$ with success probability $\geq 1/2+1/\poly(n)$ in time $T \cdot \poly(n)$ if and only if there is a bounded-error pseudodeterministic construction for $Q$ with success probability $\geq 1-2^{-n}$  in time $T \cdot \poly(n)$. 

\item[\emph{(}ii\emph{)}] There exists a zero-error pseudodeterministic construction for $Q$ with success probability $\geq 1/\poly(n)$ in time $T \cdot \poly(n)$ if and only if there is a zero-error pseudodeterministic construction for $Q$ with success probability $\geq 1-2^{-n}$ in time $T \cdot \poly(n)$.
\end{itemize}

\end{proposition}

\begin{proof}
The proof is a simple adaptation of the proof of error amplification for bounded-error and zero-error algorithms which recognize languages. Hence we just give a sketch.

We show the equivalence for bounded-error constructions first. We only need to show the forward implication, as the backward implication is trivial. By assumption, there is a bounded-error pseudodeterministic algorithm $A$ with success probability $1/2+1/\poly(n)$ such that the canonical output of $A(1^n)$ belongs to $Q_n$ for $n$ large enough. We amplify the success probability of $A$ by running it $\poly(n)$ times independently, and outputting the string that is produced most often as an output among these runs, breaking ties arbitrarily. Using standard concentration bounds, it is easy to show that this new algorithm is bounded-error pseudodeterministic with success probability at least $1-2^{-n}$, and has the same canonical output as $A$.

The proof is similar for zero-error constructions. Given a zero-error pseudodeterministic algorithm $A$ with success probability $1/\poly(n)$, we amplify the success probability by running $A$ $\poly(n)$ times independently where the number of runs is chosen large enough depending on the success probability, outputting the canonical output of $A$ if it is produced on any of the runs, and $\perp$ otherwise. Simply by independence of the runs, the canonical output of $A$ is produced by this new algorithm with success probability $1-2^{-n}$, and the new algorithm is still zero-error pseudodeterministic.
\end{proof}

As opposed to Theorem \ref{t:primes}, the algorithm from Theorem \ref{t:explicit} might generate different primes on inputs $1^\ell$ for which it is not the case that $\ell = |q_n|$ for some $n$. It is not excluded that on some inputs of this form each prime is generated with approximately equal and non-negligible probability. This can be partially addressed using the following simple proposition.

\begin{proposition}[Purifying infinitely-often pseudo-deterministic constructors] \label{p:purifying} Let $S \subseteq \mathbb{N}$ be arbitrary, and let $A$ be a bounded-error pseudodeterministic construction for $Q$ with success probability $\geq 2/3$ on every input $1^m$ with $m \in S$. Then there is a bounded-error pseudodeterministic construction $A'$ for $Q$ for which the following holds\emph{:}
\begin{itemize}
\item[\emph{(}i\emph{)}] On every input $1^m$ with $m \in S$, $A'$ outputs the same canonical solution $A(1^m)$ with high probability. 

\item[\emph{(}ii\emph{)}] On an arbitrary input $1^n$ with $n \in \mathbb{N}$, there exists a set $V'_n \subseteq \{0,1\}^n \cup \{\perp\}$ satisfying $|V'_n \cap \{0,1\}^n| \leq 1$ such that $A'(1^n) \in V'_n$ with high probability. 

\item[\emph{(}iii\emph{)}] The running time of $A'$ is at most a polynomial factor larger than the running time of $A$.
\end{itemize}
\end{proposition}

\begin{proof}
On input $1^n$, $A'$ simulates $t = n^2$ executions of $A$ on input $1^n$. Let $a_1, \ldots, a_t$ be random variables supported over $\{0,1\}^n \cup \{\perp\}$ representing the output of $A$ on each such computation. $A'$ outputs the lexicographic first string in $\{0,1\}^n$ that appears more than $0.6t$ times among $a_1, \ldots, a_t$. If no such string exists, $A'$ outputs $\perp$.

Clearly, $A'$ satisfies (\emph{iii}). It is easy to see that it also satisfies (\emph{i}) by a standard concentration bound, since on each such input $1^m$  algorithm $A$ outputs a fixed solution with probability $\geq 2/3$.  To establish (\emph{ii}), note that on an arbitrary input $1^n$, there is \emph{at most one} element $b \in \{0,1\}^n \cup \{\perp\}$ such that $\Pr[A(1^n) = b] > 1/2$. Take $V'_n = \{b, \perp\}$ if such element exists, and set $V'_n = \{\perp\}$ otherwise. In the latter case, the probability that any fixed string in $\{0,1\}^n$ appears $\geq 0.6t$ times among $a_1, \ldots, a_t$ is exponentially small, and by a union bound with high probability $A'$ outputs $\perp\,\in V'_n$. Similarly, in the former case no string $b' \neq b$ that belongs to $\{0,1\}^n$ will appear more than $\geq 0.6t$ times among $a_1, \ldots, a_t$, except with exponentially small probability. By a union bound, the output of $A'$ must be in $V'_n$ with high probability.
\end{proof}

\vspace{0.3cm}

\noindent \textbf{Complexity Classes.} We use the standard definitions for classes such as $\mathsf{P}$, $\mathsf{BPP}$, $\mathsf{RP}$, $\mathsf{ZPP}$, and $\mathsf{PSPACE}$. We will also use $\mathsf{BPE} \eqdef \mathsf{BPTIME}(2^{O(n)})$, $\mathsf{ZPE} \eqdef \mathsf{ZPTIME}(2^{O(n)})$ and $\mathsf{ESPACE} \eqdef \mathsf{DSPACE}(2^{O(n)})$. We refer to standard textbooks such as \citep{Arora-Barak09, DBLP:books/daglib/0019967} for more information about complexity theory.\\

\noindent \textbf{Boolean Circuits and Pseudorandom Generators.} We consider Boolean circuits consisting of fan-in two gates over the standard $\mathsf{B}_2$ basis, and measure circuit size by the number of wires in the circuit. The precise details will not be important here, and we refer to \citep{DBLP:books/daglib/0028687} for more background in circuit complexity. Given a size function $s\colon \mathbb{N} \rightarrow \mathbb{N}$, $\mathsf{SIZE}(s)$ is the class of properties $Q$ such that $Q_n$ has Boolean circuits of size $s(n)$, for each $n \in \mathbb{N}$.

We say that a function $G \colon \{0,1\}^\ell \to \{0,1\}^{m}$ $\varepsilon$-\emph{fools} algorithms (circuits) of running time (circuit size) $t$ if for every such algorithm (circuit) $D$,
$$
\left | \Pr_{w \sim \{0,1\}^\ell}[D(G(w)) = 1] - \Pr_{u \sim \{0,1\}^{m}} [D(u) = 1]  \right | <  \varepsilon.
$$
Otherwise, we say that $D$ $\varepsilon$-\emph{distinguishes} $G$ (on seed length $\ell$). 

We will consider a sequence $\{G_\ell\}_{\ell \in \mathbb{N}}$ of generators $G_\ell \colon \{0,1\}^{\ell} \to \{0,1\}^{m(\ell)}$, which we also view as a single function $G \colon \{0,1\}^* \to \{0,1\}^*$. The function $m(\cdot)$ is the \emph{stretch} of $G$. We say that a generator $G$ with stretch $m(\cdot)$ is \emph{quick} if on every input $w \in \{0,1\}^\ell$, $G(w)$ can be computed in time $O(m(\ell)^d)$, where $d$ is a fixed constant independent of the remaining parameters. Finally, we say that an algorithm $D$ is a $t(m)$-\emph{distinguisher} for $G$ if $D$ runs in time at most $t(m)$, and on every sufficiently large $\ell$, $D$ $1/t(m)$-distinguishes $G_\ell$.

\subsection{Technical Results}\label{ss:prelim_results}

We make use of the following ``hardness versus randomness'' results.

\begin{theorem}[Impagliazzo and Wigderson \citep{DBLP:conf/stoc/ImpagliazzoW97}]\label{t:IW_generator}
There is a function $F \colon \{0,1\}^* \times \{0,1\}^* \to \{0,1\}^*$, computable in polynomial time, for which the following holds. For every $\gamma > 0$, there exist constants $a, b \geq 1$ such that for every $r \in \mathbb{N}$, 
$$
F(\{0,1\}^{r^{a}} \times \{0,1\}^{b \log r}) \subseteq \{0,1\}^r,
$$ 
and if $T_h \in \{0,1\}^{r^a}$ is the truth-table of a function $h \colon \{0,1\}^{a \log r} \to \{0,1\}$ of circuit complexity $\geq r^{\gamma a}$, then 
$$
G_h \colon \{0,1\}^{b \log r} \to \{0,1\}^r, \;\; \text{given by}\;\; G_h(w) \eqdef F(T_h,w),
$$ 
is a \emph{(}quick\emph{)} pseudorandom generator that $(1/r)$-fools size-$r$ circuits.
\end{theorem}

The next generator is based on a downward-self-reducible and self-correctable $\mathsf{PSPACE}$-complete language $L^\star$ described in \citep{DBLP:journals/cc/TrevisanV07}, and the corresponding result can be seem as a uniform ``low-end'' version of Theorem \ref{t:IW_generator}. The following formulation is sufficient for our purposes.

\begin{theorem}[Impagliazzo and Wigderson \citep{DBLP:journals/jcss/ImpagliazzoW01}, Trevisan and Vadhan \citep{DBLP:journals/cc/TrevisanV07}] \label{t:uniform_hard_rand}
For every integers $b, c \geq 1$, there exists an integer $d \geq 1$ and a function $G^\star \colon \{0,1\}^* \to \{0,1\}^*$ with restrictions 
$$G_\ell^\star \colon \{0,1\}^\ell \to \{0,1\}^{m(\ell)},\quad \text{where}\;\;\;m(\ell) = \ell^b,$$ such that $G^\star$ can be computed in time $O(m(\ell)^d) = \mathsf{poly}(\ell)$ when given oracle access to $L^\star_{\leq \ell}$, and the following holds.  If the output of $G^\star_\ell$ can be $(1/m(\ell)^c)$-distinguished from random for every large enough $\ell \in \mathbb{N}$ by an algorithm running in time $O(m(\ell)^c)$, then $\mathsf{PSPACE} \subseteq \mathsf{BPP}$.
\end{theorem}

\section{Pseudodeterministic Constructions}
\label{s:proof_theorems}

\subsection{Conditional Constructions}
\label{s:cond_cons}

\begin{lemma}
\label{l:cond_discset}
If there is an $\varepsilon > 0$ and a Boolean function $h \in \mathsf{BPE}$ \emph{(}resp. $h \in \mathsf{ZPE}$\emph{)} that requires circuits of size $2^{\varepsilon m}$ on all large enough input lengths $m \in \mathbb{N}$, then for each constant $c > 0$, there is a bounded-error \emph{(}resp. zero-error\emph{)} pseudodeterministic algorithm $A$ running in polynomial time such that $\{A(1^n)\}$ is a $(1/n^c)$-discrepancy set family for every property $Q$ decidable in deterministic time $n^c$.  
\end{lemma}

\begin{proof}
We give the proof for the case that $h \in \mathsf{BPE}$; the case that $h \in \mathsf{ZPE}$ is exactly analogous. Assume without loss of generality that the probabilistic exponential-time machine $N$ computing $h$ has failure probability $< 2^{-2m}$ on any input of length $m$.

We define the behaviour of pseudodeterministic algorithm $A$ on input $1^n$ using Theorem \ref{t:IW_generator} as follows. We apply this result with constants $a$ and $b$ corresponding to the choice $\gamma = \varepsilon$. Let $r = n^{2c}$. The algorithm $A$ first computes the truth table $T_h$ of the function $h$ on input size $a \log r$ by evaluating $h$ on every input of size $a \log r$ using the probabilistic exponential-time machine $N$. (We assume here that $r$ is a power of $2$; if not, the same argument works by setting $r$ to be the smallest power of $2$ larger than $n^{2c}$.) Note that by a union bound, $A$ computes the entire truth table correctly with probability $1-o(1)$, and runs in time $\poly(n)$. $A$ then computes $F(T_h, y)$ for each string $y \in \{0,1\}^{b \log r}$, using the polynomial-time algorithm for $F$ guaranteed by Theorem \ref{t:IW_generator}, and forms a multi-set $H'_n \subseteq \{0,1\}^{n^{2c}}$ composed of the union of these strings. It then forms the (multi-) set $H_n \eqdef \{\mathsf{left}_{n}(u) \mid u \in H'_n\}$, and outputs $H_n$. 

Note that the same set $H_n$ is output whenever the correct truth table $T_h$ is computed, which happens with probability $1-o(1)$ as argued before. Hence $A$ is a pseudodeterministic algorithm with success probability $1-o(1)$. Let $H_n$ be the canonical output of $A(1^n)$.

It remains to argue that $H_n$ is a $1/n^c$-discrepancy set for $Q_n$ for large enough $n$, where $Q$ is any property decidable in time $n^c$ for $n$ large enough. Let $M$ be a deterministic machine deciding $Q$ with a running time of this form, and let $D_n$ be the Boolean circuit of size $n^{2c}$ obtained by translating $M$'s computation on inputs of length $n$ into a circuit, as in the proof of the Cook-Levin theorem. Let $C_n$ be the circuit of size $n^{2c}$ on input $x$ of size $n^{2c}$, which simulates $D_n$ on $\mathsf{left}_n(x)$ and outputs the answer. Finally, let $n$ be large enough that $h$ requires circuits of size at least $m^{2ac \varepsilon}$ on inputs of size $2ac \log m$ for any $m \geq n$ that is a power of two. By Theorem \ref{t:IW_generator}, $G_h\colon \{0,1\}^{b \log r} \to \{0,1\}^r$ $1/n^{2c}$-fools $C_n$. It follows that $H_n$ is a $1/n^c$-discrepancy set for $Q_n$ for $n$ large enough, using the facts that $C_n$ accepts a string $u$ of length $n^{2c}$ if and only if $D_n$ accepts $\mathsf{left}_n(u)$, and that $D_n$ accepts a string $z \in \{0,1\}^n$ if and only if $z \in Q_n$.
\end{proof}

\begin{remark}\label{r:nonuniform} We note that the conclusion of Lemma \emph{\ref{l:cond_discset}} holds even for properties decidable by Boolean circuits of size $\leq n^c$, since we do not take advantage of uniformity in the argument above.
\end{remark}

\begin{corollary}
\label{c:cond_construction}
Let $Q$ be any property that is easy and dense. If there is an $\varepsilon > 0$ and a Boolean function $f \in \mathsf{BPE}$ \emph{(}resp. $f \in \mathsf{ZPE}$\emph{)} that requires circuits of size $2^{\varepsilon m}$ on all large enough input lengths $m \in \mathbb{N}$, then there is a bounded-error \emph{(}resp. zero-error\emph{)} pseudodeterministic construction for $Q$ in polynomial time.
\end{corollary}

\begin{proof}

Again we give the proof for the case that $f \in \mathsf{BPE}$; the case that $f \in \mathsf{ZPE}$ is exactly analogous. Let $Q$ be any property that is easy and dense. Let $M$ be a polynomial-time Turing machine deciding $Q$. Let $c$ be a constant chosen large enough such that $Q_n$ is $1/n^c$-dense and $M$ runs in time at most $n^c$ for $n$ large enough. Using Lemma \ref{l:cond_discset}, we have that there is a bounded-error pseudodeterministic algorithm $A$ running in polynomial time such that the canonical output $H_n$ of $A(1^n)$ is a $1/n^c$-discrepancy set for $Q_n$. Hence, by the density condition on $Q$, we have that $H_n$ is a hitting set for $Q_n$ for $n$ large enough.

We define a bounded-error pseudodeterministic algorithm $A_Q$ running in polynomial time, whose canonical output on input $1^n$ belongs to $Q_n$ for $n$ large enough.  $A_Q(1^n)$ first simulates $A(1^n)$ to obtain a subset $S_n$ of $\{0,1\}^n$; if the output of $A(1^n)$ is not such a subset, it outputs an arbitrary string. With probability $1-o(1)$, $S_n$ is the canonical output of $A(1^n)$, which is the hitting set $H_n \subseteq \{0,1\}^n$ of size $\poly(n)$. $A_Q$ then orders the strings in $S_n$ in lexicographic order to obtain a list $y_1, y_2, \ldots , y_m$, where $m = \poly(n)$, and each $y_i$, $ i = 1, \ldots , m$, is an $n$-bit string. $A_Q$ simulates $M$ on each of the $y_i$'s, in order, until it finds a $y_i$ on which $M$ accepts. If it finds such a $y_i$, it outputs $y_i$, otherwise it rejects.

For $n$ large enough, let $z_n$ be the smallest element of $H_n$ in lexicographic order which belongs to $Q_n$. Since $H_n$ is a hitting set for $Q_n$ for $n$ large enough, such a string $z_n$ exists. It is easy to see that $A_Q$ outputs $z_n$ with probability $1-o(1)$, and hence that $z_n$ is the canonical output of $A_Q(1^n)$. This is because $S_n = H_n$ with probability $1-o(1)$, and whenever $S_n = H_n$, the string $z_n$ is output by $A_Q$. Clearly $A_Q$ runs in time $\poly(n)$ and is a pseudodeterministic algorithm with success probability $1-o(1)$ whose canonical output is in $Q_n$ for $n$ large enough.
\end{proof}

The proof of Corollary \ref{c:cond_construction} can easily be adapted to give a bounded-error pseudodeterministic construction in polynomial time even when the condition on easiness of $Q$ is relaxed to $Q \in \mathsf{BPP}$, and a zero-error pseudodeterministic construction in polynomial time even when the condition on easiness of $Q$ is relaxed to $Q \in \mathsf{ZPP}$. There is also a natural tradeoff between the hardness assumption and the running time of the corresponding pseudodeterministic construction, which can be obtained by using a generalization of Theorem \ref{t:IW_generator} giving a hardness-randomness tradeoff based on circuit lower bounds varying from superpolynomial to exponential.

Note that the easiness assumption on $Q$ is used twice to obtain Corollary \ref{c:cond_construction}, the first time in the proof of Lemma \ref{l:cond_discset} to obtain pseudodeterministic algorithms outputting discrepancy sets based on the hardness assumption, and the second time in the proof of Corollary \ref{c:cond_construction} to obtain a pseudodeterministic construction for $Q$ from the discrepancy set. The first use of the assumption can be eliminated by using a stronger hardness assumption with respects to circuits that have oracle access to $Q$, but the second use seems essential.

\subsection{Unconditional Constructions that Work Infinitely Often}
\label{s:uncond_cons}

By using in addition uniform hardness-randomness tradeoffs, we obtain {\it unconditionally} that there are pseudodeterministic constructions for easy dense sets in polynomial time, or else there are deterministic constructions for easy dense sets in subexponential time which work infinitely often. In fact, we obtain the following stronger result stating that in either case, there is a {\it generic} hitting set family that works simultaneously for {\it all} easy dense sets where the easiness and density parameters are fixed in advance.

\begin{theorem}[Restatement of Theorem \ref{t:generic}]
\label{t:uncond_hittingset}
Let $c > 0$ be any constant. At least one of the following is true\emph{:}
\begin{enumerate}

\item[\emph{1.}] There is a pseudodeterministic algorithm $A$ running in polynomial time such that $\{A(1^n)\}$ is a hitting set family for every $1/n^c$-dense property $Q$ decidable in time $n^c$.

\item[\emph{2.}] For each $\varepsilon > 0$, there is a deterministic algorithm $B_{\varepsilon}$ running in time $O(2^{n^{\varepsilon}})$, such that $\{B_{\varepsilon}(1^n)\}$ is an i.o.~hitting set family for every $1/n^c$-dense property $Q$ decidable in time $n^c$.
\end{enumerate}
\end{theorem}

We now give the proof of Theorem \ref{t:uncond_hittingset}. Informally, the proof is a ``win-win-win'' argument using uniform hardness-randomness tradeoffs. We will describe two candidate hitting set families $\mathcal{H}^{\mathsf{easy}}$ and $\mathcal{H}^{\mathsf{hard}}$ computable in sub-exponential time, the first based on the ``easy witness'' method of Kabanets \citep{DBLP:journals/jcss/Kabanets01}, and the second based on the uniform hardness-randomness tradeoffs of Impagliazzo-Wigderson and Trevisan-Vadhan \citep{DBLP:journals/jcss/ImpagliazzoW01, DBLP:journals/cc/TrevisanV07}. We will show that if there is a sufficiently easy dense property $Q$ such that $\mathcal{H}^{\mathsf{easy}}$ is not an i.o.~hitting set family for $Q$, then $\mathsf{BPP} = \mathsf{ZPP}$. We will also show that if there is a sufficiently easy dense property $Q$ such that $\mathcal{H}^{\mathsf{hard}}$ is not an i.o.~hitting set family for $Q$, then $\mathsf{PSPACE} = \mathsf{BPP}$. Thus, in each of these cases, we either ``win'' for each sufficiently easy dense property $Q$ by computing an i.o.~hitting set family in sub-exponential time, or we have a complexity collapse. Finally, we show how to ``win'' in the case that both complexity collapses occur, by using Lemma \ref{l:cond_discset}. The win in this case is in the form of a polynomial-time pseudodeterministic algorithm which outputs a hitting set family for each sufficiently easy dense property. When either of the first two wins occur, item 2 of Theorem \ref{t:uncond_hittingset} holds, and when the third win occurs, item 1 holds.

Fix a $c > 0$ as in the statement of Theorem \ref{t:uncond_hittingset}, and let $\varepsilon > 0$ be any constant. The candidate hitting set families $\mathcal{H}^{\mathsf{easy}}$ and $\mathcal{H}^{\mathsf{hard}}$ depend on $\varepsilon$, but we work with a fixed arbitrarily small $\varepsilon > 0$ through the proof, and therefore do not need to formalize this dependence.

We first define the candidate hitting set family $\mathcal{H}^{\mathsf{easy}}$. Consider the set $\mathcal{C}_n$ of all Boolean circuits $C$ on $\lceil \log n \rceil$ input variables and of size at most $n^\delta$, where $\delta \eqdef \varepsilon/10$. Each circuit $C \in \mathcal{C}_n$ computes a Boolean function $f_C \colon \{0,1\}^{\lceil \log(n) \rceil} \to \{0,1\}$, and its truth-table $\mathsf{tt}(f_C)$ is a string of length $\geq n$. We consider the following  family $\mathcal{H}^{\mathsf{easy}} = \{H_n^{\mathsf{easy}}\}$  obtained from $\mathcal{C}_n$:
$$
H_n^{\mathsf{easy}} \eqdef \{\mathsf{left}_n(\mathsf{tt}(f_C)) \mid C \in \mathcal{C}_n\} \subseteq \{0,1\}^n.
$$

By our choice of $\delta = \varepsilon/10$ and a standard upper bound on the number of small circuits (cf. \citep{Arora-Barak09}), there are at most $O(2^{n^{2\delta}}) \leq O(2^{n^{\varepsilon/2}})$ strings in $\mathcal{H}_n^\mathsf{easy}$. Furthermore, an ordered list containing all such strings can be printed in time $O(2^{n^{\varepsilon}})$. 

The following lemma is key to our analysis of $\mathcal{H}^{\mathsf{easy}}$.

\begin{lemma}\label{l:A1_fails}
If there is a $1/n^c$-dense property $Q$ decidable in deterministic time $n^c$ such that $\mathcal{H}^\mathsf{easy}$ is not an i.o.~hitting set family for $Q$, then $\mathsf{BPP} \subseteq \mathsf{ZPP}$.
\end{lemma}

\begin{proof}
Let $L \in \mathsf{BPP}$, and $V_L$ be a polynomial time verifier that decides $L$ in time $n^{c_L/2}$, where $c_L \geq 1$. In other words, if $x \in L$ then $\Pr_{y}[V_L(x,y) = 1] \geq 2/3$, while if $x \notin L$ we have $\Pr_{y}[V_L(x,y) = 1] \leq 1/3$. By assumption, there is a $1/n^c$-dense property $Q$ decidable in deterministic time $n^c$ such that $H_n^\mathsf{easy} \cap Q_n = \emptyset$ whenever $n \geq n_0$, where $n_0$ is a fixed constant. 
Let $V$ be a deterministic Turing machine running in time $n^c$ for $n$ large enough, and deciding $Q$.

The following zero-error algorithm $B_1$ decides $L$ in polynomial time. Let $x \in \{0,1\}^n$ be the input string of $B_1$. If $n < n_0$, $B_1$ outputs the correct answer by storing it on its code. Otherwise, consider the Boolean circuit $D_x$ of size at most $n^{c_L}$ obtained from $V_L$ by fixing its first input to $x$. In order to decide $L$ on $x$ it is enough to estimate the acceptance probability of $D_x$ with additive error at most $1/10$.

Let $m \eqdef n^{c_L}$, $\gamma \eqdef \delta$, and assume from now on that $n \geq n_0$. Let $F$ be the polynomial time computable function from Theorem \ref{t:IW_generator}, and $a,b \geq 1$ be the respective constants for our choice of $\gamma$. Finally, set $\ell \eqdef a \cdot \log(r) = a \cdot c_L \cdot \log(n)$, and $m \eqdef 2^\ell = n^{a \cdot c_L} \leq \poly(n)$.

Algorithm $B_1$ samples $k \eqdef m^{10c}$ independent uniformly distributed strings $z_1, \ldots, z_k \sim \{0,1\}^{m}$. Let $z_i$ be the first string on this list such that $V(z_i) = 1$, if such string exists. Otherwise, $B_1$ aborts its computation. Using the hypothesis that $Q$ is $n^{-c}$-dense, we get that whenever $n$ is sufficiently large, algorithm $B_1$ succeeds with high probability in finding a string $z_i$ of this form. By redefining $n_0$, we can assume without loss of generality that $B_1$ succeeds whenever $n \geq n_0$.

Since $m \geq n \geq n_0$ and $z_i \in Q$, it follows from our previous discussion that $z_i \notin H_{m}^\mathsf{easy}$. In other words, $z_i$ is \emph{not} the leftmost segment of the truth-table of a Boolean function $f \colon \{0,1\}^{a \cdot \log r} \to \{0,1\}$ of circuit complexity at most $m^{\delta} = r^{\gamma a}$. 

Let $h \colon \{0,1\}^{a \log m} \to \{0,1\}$ be the Boolean function encoded by the string $z_i$ completed with zeroes until the next power of two. The previous paragraph implies that $h$ has circuit complexity greater than $r^{\gamma a}$.  It follows from Theorem \ref{t:IW_generator} that $G_h \colon \{0,1\}^{b \log r} \to \{0,1\}^r$ is a generator that $(1/10)$-fools circuits of size at most $r = n^{c_L}$. Since $F$ and $h$ are efficiently computable in $n$, $B_1$ can compute $G_h(w) = F(z_i,w)$ on every input $w \in \{0,1\}^{a \log r}$ in time $\mathsf{poly}(n)$. Moreover, the seed length of $G_h$ is $O(\log n)$. By standard methods, i.e. by trying all possible seeds of $G_h$ and taking a majority vote using $D_x$, it follows that $B_1$ can efficiently decide whether $x$ is in $L$. In addition, it is clear that $B_1$ decides $L$ with zero-error, since it aborts (with a small probability) when a good string $z_i$ is not found. Consequently, $L \in \mathsf{ZPP}$, which completes the proof of Lemma \ref{l:A1_fails}. 
\end{proof}

\vspace{0.2cm}

Next we define the candidate hitting set family $\mathcal{H}^{\mathsf{hard}}$. Assume that the $\mathsf{PSPACE}$-complete language $L^\star$ from Theorem \ref{t:uniform_hard_rand} can be decided in space $O(n^a)$, where $a \in \mathbb{N}$. Let $\delta \eqdef \varepsilon/(10a)$. Consider the function $G^\star$ obtained from an application of Theorem \ref{t:uniform_hard_rand} with stretch parameter $b = 1/\delta$ and the parameter $c$ set as in the statement of Theorem \ref{c:uncond_cons}. This gives a sequence $\{G^\star_\ell\}_{\ell \in \mathbb{N}}$ of functions computable with oracle access to $L^\star_{\leq \ell}$ in time $\mathsf{poly}(\ell)$, where each $G^\star_\ell \colon \{0,1\}^\ell \to \{0,1\}^{\ell^b}$. Consider a seed function $s(n) \eqdef \lceil n^\delta \rceil$, and let $G_{n} \eqdef G^\star_{s(n)}$. Since $b = 1/\delta$, we have 
$$
G_{n} \colon \{0,1\}^{\lceil n^\delta \rceil} \to \{0,1\}^{ \geq n}.
$$
In addition, $G_n$ can be computed in time $O(2^{n^{\varepsilon/5}})$ without access to an oracle. This is because the oracle answers to $L_{\leq n^\delta}$ can be computed within this running time due to our choice of parameters and the fact that bounded-space algorithms run in time at most exponential. 

Each set $H_n^{\mathsf{hard}}$ is obtained from $G_n$ as follows:
$$
H_n^{\mathsf{hard}} \eqdef \{\mathsf{left}_n(G_n(w)) \mid w \in \{0,1\}^{s(n)}\}.
$$

The following lemma is key to our analysis of $\mathcal{H}^{\mathsf{hard}}$.

\begin{lemma}\label{l:A2_fails}
If there is a $1/n^c$-dense property $Q$ decidable in deterministic time $n^c$ such that $\mathcal{H}^\mathsf{hard}$ is not an i.o.~hitting set family for $Q$, then $\mathsf{PSPACE} \subseteq \mathsf{BPP}$.
\end{lemma}

\begin{proof}
If $\mathcal{H}^{\mathsf{hard}}$ is not an i.o.~hitting set family for $Q$, it must be the case that $H_n^{\mathsf{hard}} \cap Q_n = \emptyset$ for every large enough $n$. Equivalently, $\mathsf{left}_n(G^\star_{s(n)}(\{0,1\}^{s(n)}))\, \cap \, Q_n = \emptyset$, where $s(n) = n^\delta$. Since $s(\cdot)$ is surjective as a function in $\mathbb{N} \to \mathbb{N}$, for every large enough $\ell$, 
$$
G^\star_\ell(\{0,1\}^\ell)\, \cap \, Q_{\ell^b} = \emptyset.
$$
On the other hand, $Q$ is a $n^{-c}$-dense property, and for large enough $\ell$,
$$
\Pr_{y \sim \{0,1\}^{\ell^b}}[y \in E_{\ell^b}]\;=\;\Pr_{y \sim \{0,1\}^{\ell^b}}[V(y)=1] \;\geq \; \ell^{-bc} \; = \; m^{-c},
$$
using $m = \ell^b$ as in Theorem \ref{t:uniform_hard_rand}. Furthermore, by assumption there is a deterministic Turing machine $V$ running in time at most $m^c$ on inputs of length $m$ for large enough $m$, and deciding $Q$. Therefore, it can be used as a distinguisher against $G^\star$, matching the parameters of Theorem \ref{t:uniform_hard_rand}. Thus $\mathsf{PSPACE} \subseteq \mathsf{BPP}$, completing the proof of Lemma \ref{l:A2_fails}.
\end{proof}

\vspace{0.2cm}

We complete the proof of Theorem \ref{t:uncond_hittingset} using a third and final application of the hardness versus randomness paradigm. If $\mathcal{H}^{\mathsf{easy}}$ is an i.o. hitting set family for every $n^{-c}$-dense property $Q$ decidable in deterministic time $n^c$, the second item of Theorem \ref{t:uncond_hittingset} holds by letting $B_{\varepsilon}$ be the deterministic algorithm which, on input $1^n$, outputs $H_n^{\mathsf{easy}}$ in time $2^{O(n^{\varepsilon})}$. If not, then we have $\mathsf{BPP} = \mathsf{ZPP}$ by Lemma \ref{l:A1_fails}. If $\mathcal{H}^{\mathsf{hard}}$ is an i.o.~hitting set family for every $n^{-c}$-dense property $Q$ decidable in deterministic time $n^c$, then again, the second item of Theorem \ref{t:uncond_hittingset} holds by letting $B_{\varepsilon}$ be the deterministic algorithm which, on input $1^n$, outputs $H_n^{\mathsf{hard}}$ in time $2^{O(n^{\varepsilon})}$. If not, then we have $\mathsf{PSPACE} = \mathsf{BPP}$ by Lemma \ref{l:A2_fails}. Thus, if the second item of Theorem \ref{t:uncond_hittingset} fails to hold, we have the complexity collapse $\mathsf{PSPACE} = \mathsf{ZPP}$. This facilitates a zero-error pseudodeterministic construction of an unconditional hitting set in polynomial time as follows.

First, it follows by direct diagonalization that there is a language computed in $\mathsf{DSPACE}(2^{O(m)})$ that requires circuits of size $\geq 2^{m/2}$ for every large $m$. From $\mathsf{PSPACE} \subseteq \mathsf{ZPP}$, a standard padding argument implies that this language can be computed in $\mathsf{ZPTIME}(2^{O(m)})$. In other words, there is a function $h\colon \{0,1\}^* \to \{0,1\}$ in $\mathsf{ZPTIME}(2^{O(m)})$ that for every large $m$ cannot be computed by circuits of size $\leq 2^{m/2}$. 

Now, by setting $\varepsilon = 1/2$ in Lemma \ref{l:cond_discset}, we have that there is a pseudodeterministic algorithm $A$ running in polynomial time such that $\{A(1^n)\}$ is a $1/n^c$-discrepancy set family for every property $Q$ decidable in deterministic time $n^c$, and consequently a hitting set family for every $1/n^c$-dense property $Q$ decidable in deterministic time $n^c$. This completes the proof of Theorem \ref{t:uncond_hittingset}. \qed 
   
\vspace{0.2cm}   

\begin{remark}
Observe that the argument presented above provides a stronger \emph{discrepancy set} in the first case of Theorem \emph{\ref{t:uncond_hittingset}}. While this is not needed in our applications, it might be helpful elsewhere.
\end{remark}

The generic Theorem \ref{t:uncond_hittingset} can be used to show unconditionally the existence of certain kinds of explicit constructions for easy dense properties.
 
\begin{corollary}
\label{c:uncond_cons}
Let $Q \subseteq \{0,1\}^*$ be any easy dense property. Then for each $\varepsilon > 0$, there is an i.o.~zero-error pseudodeterministic construction for $Q$ in time $O(2^{n^{\varepsilon}})$. 
\end{corollary}

To establish Corollary \ref{c:uncond_cons}, note that if the first item of Theorem \ref{t:uncond_hittingset} holds, a zero-error pseudodeterministic construction for $Q$ in polynomial time follows exactly as Corollary \ref{c:cond_construction} follows from Lemma \ref{l:cond_discset}. On the other hand, if the second item of Theorem \ref{t:uncond_hittingset} holds, for every $\varepsilon > 0$, there is an i.o.~deterministic construction for $Q$ in time $O(2^{n^{\varepsilon}})$, just by computing the i.o.~hitting sets for $Q$ and outputting the lexicographically first element of the hitting set in $Q$, if such an element exists, and an arbitrary fixed string otherwise. Thus, in either case, for every $\varepsilon > 0$ there is an i.o.~zero-error pseudodeterministic construction for $Q$ in time $O(2^{n^{\varepsilon}})$.

We could trade off the parameters of items (1) and (2) of Theorem \ref{t:uncond_hittingset} to obtain a stronger bound on the running time of the construction in Corollary \ref{c:uncond_cons} by using a more general version of Theorem \ref{t:IW_generator}, but we do not pursue this direction here, as a polynomial-time bound for the construction does not appear to be provable using such an approach.

Corollary \ref{c:uncond_cons} also has a non-constructive element. We know that for any easy dense $Q$, there is an i.o.~zero-error pseudodeterministic construction, but we are unable to say explicitly what this construction is, as we do not which of the worlds {\bf PSEUDO} or {\bf SPARSE} we live in. Also, similarly to Theorem \ref{t:uncond_hittingset}, we do not have any information on the set of input lengths for which the construction works, except that it is infinite. This can be a somewhat unsatisfactory situation where explicit constructions are concerned, and we show how to address both issues in Section \ref{s:explicit}.

\begin{corollary}[Restatement of Theorem \ref{t:primes}]
\label{c:uncond_primes}
For each $\varepsilon > 0$, there is an i.o.~zero-error pseudodeterministic construction for $\mathsf{Primes}$ in time $O(2^{n^{\varepsilon}})$.
\end{corollary}

Corollary \ref{c:uncond_primes} follows from Corollary \ref{c:uncond_cons} because $\mathsf{Primes}$ is $1/\poly(n)$-dense by the Prime Number Theorem, and is in deterministic polynomial time by the Agrawal-Kayal-Saxena algorithm \cite{Agrawal02primesis}.

Using the ideas of the proof of Theorem \ref{t:uncond_hittingset}, we can partially answer a question of the Polymath 4 Project on generating primes. The following question was posed there: does $\mathsf{BPP} = \mathsf{P}$ imply a polynomial-time algorithm for generating primes? We consider a much weaker assumption, namely that $\mathsf{ZPP} \subseteq \io \mathsf{DTIME}(2^{n^c})$ for some fixed constant $c$. Under this assumption, we show that there is a subexponential-time deterministic algorithm for generating infinitely many primes.

\begin{theorem}
\label{t:cond_detprimes}
If there is a $c \geq 1$ such that $\mathsf{ZPP} \subseteq \io \mathsf{DTIME}(2^{n^c})$, then for each $\varepsilon > 0$ there is a deterministic algorithm $A$ running in time $O(2^{n^{\varepsilon}})$ such that for infinitely many $n$, $A(1^n)$ is an $n$-bit prime.
\end{theorem}

\begin{proof}
We just give a sketch, as the argument largely relies on the proof of Theorem \ref{t:uncond_hittingset}. The proof of Theorem \ref{t:uncond_hittingset} establishes that if $\mathsf{PSPACE} \neq \mathsf{ZPP}$, there is a subexponential-time algorithm for generating an i.o.~hitting set family for sufficiently easy dense properties, and hence an i.o.~deterministic construction for primes in subexponential time. 

Now consider the case that $\mathsf{PSPACE} = \mathsf{ZPP}$. Either we have that $\mathsf{EXP}$ admits polynomial-size circuits, or it does not. In the latter case, by using Theorem \ref{t:IW_generator}, we again have subexponential-time generatable hitting sets for properties that are easy and dense, and hence an i.o.~deterministic construction for primes in subexponential time. In the former case, by standard Karp-Lipton theorems \citep{DBLP:conf/stoc/KarpL80}, $\mathsf{EXP} = \mathsf{PSPACE}$, and hence $\mathsf{EXP} = \mathsf{ZPP}$. But in this case, the assumption that $\mathsf{ZPP} \subseteq \io \mathsf{DTIME}(2^{n^c})$ gives a contradiction to the almost-everywhere deterministic time hierarchy theorem. This concludes the argument.
\end{proof}

As an example of how the generic nature of Theorem \ref{t:uncond_hittingset} is useful, consider the question of generating incompressible strings. Let $f \colon \{0,1\}^* \to \{0,1\}^*$ be a polynomial time computable function. We say that $f$ is a \emph{compression scheme} if $f$ is an injective function and on every $x \in \{0,1\}^*$, $|f(x)| \leq |x|$. We use $\mathrm{I}^f \eqdef \{x \in \{0,1\}^* : \,|f(x)| \geq |x|-1\}$ to denote the set of $f$-\emph{almost incompressible} strings.

\begin{corollary}
\label{c:incomp}
Let $f \colon \{0,1\}^* \to \{0,1\}^*$ be an arbitrary polynomial time compression scheme.  For every $\varepsilon > 0$, there is a zero-error pseudodeterministic algorithm $A$ running in time $O(2^{n^\varepsilon})$ such that on infinitely many values of $n$, $A(1^n)$ outputs an $f$-almost incompressible string of length $n$.
\end{corollary}

\begin{proof}
Observe that $\mathrm{I}^f$ can be decided in polynomial time, since $f$ is a polynomial time computable function. Let $\mathrm{I}^f_n \eqdef \mathrm{I}^f \cap \{0,1\}^n$. By a simple counting argument that uses the injectivity of $f$, for every $n \in \mathbb{N}$ we have $|\mathrm{I}^f_n|/2^n \geq 1/2$. Consequently, $\mathrm{I}^f$ is a dense language. The result now follows immediately from Corollary  \ref{c:uncond_cons}.
\end{proof}

\subsection{Pseudodeterministically Sampling Distributions}\label{ss:samplable}

In order to state the main result of this section we will need the following additional definition.\\
\vspace{-0.2cm}

\noindent \textbf{Samplable Distributions.} Let $\mathfrak{D} = \{\mathcal{D}_n\}_{n \in \mathbb{N}}$ be an ensemble of probability distributions, where each $\mathcal{D}_n$ is supported over $\{0,1\}^*$. For a string $a \in \{0,1\}^*$, we let $\mathcal{D}_n(a)$ denote the probability of $a$ under $\mathcal{D}_n$. We say that $\mathfrak{D}$ is \emph{polynomial time samplable} if there exists a polynomial time computable function $g \colon 1^* \times \{0,1\}^* \to \{0,1\}^* \cup \{\esymb\}$ for which the following holds:
\begin{itemize}
\item There exists an integer $c \geq 1$ such that for every $n$, $g(1^n, \{0,1\}^{n^c}) \subseteq \mathsf{Support}(\mathcal{D}_n) \cup \{\esymb\}$. 
\item There exists $k \geq 1$ such that for every $n$, $\Pr_{w \sim U_{n^c}}[g(1^n, w) = \,\esymb\,] \leq 1 - 1/n^k$.
\item For every $n$ and $a \in \{0,1\}^*$, $\mathcal{D}_n(a) = \Pr_{w \sim U_{n^c}}[g(1^n, w) = a \mid g(1^n, w) \neq \,\esymb\,]$. In other words, $\mathcal{D}_n$ and $g(1^n, U_{n^c})$ are equally distributed conditioned on the output of $g$ being different from the error symbol ``$\esymb$''.
\end{itemize}

\noindent Observe that the support of each $\mathcal{D}_n$ is not required to be efficiently computable.\\

We define the computational problem of generating a canonical sample from $\mathfrak{D}$ in the natural way. In other words, given $1^n$, the algorithm must produce a ``canonical'' string in the support of $\mathcal{D}_n$. Since every dense and easy property is polynomial time samplable, this setting provides a generalization of the explicit construction problem associated with such properties. (Observe that the difficulty of producing a canonical sample comes from the fact that the polynomial time function $g$ can fail with high probability. Of course, if $g$ never fails, it is enough to output, say, $g(1^n,0^{n^c})$.)

\begin{theorem}[Pseudodeterministic Samplers in $\mathtt{i.o.}$Subexponential Time]\label{t:canonical_sampler}
Let $\mathfrak{D}$ be a polynomial time samplable ensemble. Then $\mathfrak{D}$ admits subexponential time randomized algorithms that output a canonical sample from $\mathcal{D}_n$ for infinitely many values of $n$.
\end{theorem}

\noindent (We stress that the pseudodeterministic sampler from Theorem \ref{t:canonical_sampler} is not zero-error: it will output on every input $n$ where it succeeds a fixed sample in the support of $\mathcal{D}_n$ with very high probability, but it might output a different sample in $\mathcal{D}_n$ with negligible probability. On the input lengths where it fails, it will never output a sample.)

\begin{proof}
The argument is a reduction to Theorem \ref{t:uncond_hittingset}, or more precisely, to Corollary \ref{c:uncond_cons}. First, we partition $\mathbb{N}^+$ into infinitely many sets $S_i \eqdef \{i^c, \ldots, (i+1)^c - 1\}$, where $c$ is the positive integer provided by the definition of the polynomial time samplable ensemble $\mathfrak{D}$, and $i \in \mathbb{N}^+$. Let $g_\mathfrak{D}$ be the polynomial time computable function associated with $\mathfrak{D}$. For convenience, given a string $x \in \{0,1\}^m$ and an interval $S \subseteq [m]$, we use $x_S$ to denote the substring of $x$ with coordinates in $S$. Next we define a property $Q \subseteq \{0,1\}^*$ via the following polynomial time algorithm $A_Q$. On an input $x \in \{0,1\}^m$, let $i$ be the unique positive integer such that $m \in S_i$. $A_Q$ computes $\alpha_x = g(1^{i},x_{[1, i^c]})$, and accepts $x$ if and only if $\alpha_x \neq \esymb$. 

We claim that $Q$ is dense and efficiently computable. First, observe that $m \geq i^c$, since $m \in S_i$. Consequently, $i \leq m^{1/c}$, and since $g$ is polynomial time computable, so is $A_Q$. Similarly, using that $g$ has a non-$\esymb$ output with inverse polynomial probability and that $i$ and $m$ are polynomially related, it follows that $Q$ is a dense property. 

Using the claim from the previous paragraph and Corollary \ref{c:uncond_cons}, it follows that for every $\varepsilon > 0$, $Q$ admits a zero-error pseudodeterministic constructor $B_\varepsilon$ running in zero-error time $\leq 2^{m^\varepsilon}$ that succeeds on infinitely many values of $m$. We argue next that each $B_\varepsilon$ can be used to sample a canonical string from $\mathcal{D}_n$ for infinitely many values of $n$, where $n$ and $m$ are polynomially related.

Fix an arbitrary algorithm $B_\varepsilon$ as above. Since $B_\varepsilon$ is a zero-error algorithm, we can assume without loss of generality and increase of running time that on every input $1^m$ where it succeeds, its output is $\perp$ with probability at most $2^{-m}$. We define from $B_\varepsilon$ a randomized algorithm $A_\varepsilon$ that computes as follows. Given an input of the form $1^n$, we run $B_\varepsilon(1^m)$ on every $m \in S_n$ using independent random strings of length $\leq 2^{m^\varepsilon}$. Let $m'$ be the smallest element in $S_n$ such that $B_\varepsilon(1^{m'}) \neq\,\perp$ among such executions, if it exists, and let $z_{m'} \in \{0,1\}^{n^c}$ be the random variable denoting the first $n^c$ bits of the output of $B_\varepsilon$ for $m'$. We set the output of $A_\varepsilon$ to $g(1^n, z_{m'})$ in this case, and to $\perp$ if no such $m'$ exists. This completes the construction of $A_\varepsilon$ from $B_\varepsilon$ and $\mathfrak{D}$.

First, observe that $A_\varepsilon$ is a randomized algorithm that runs in time $\leq 2^{O(n^{c\varepsilon})}$, since each $m \in S_n$ is of size $O(n^{c})$, there are polynomially many simulations of $B_\varepsilon$ on input $1^n$, each running in time at most $2^{O(m^\varepsilon)}$, and $g$ is polynomial time computable. In order to argue the correctness of $A_\varepsilon$, recall that each $B_\varepsilon$ succeeds on infinitely many values of $m$, and that the sets $S_i$ form a partition of $\mathbb{N}^+$ into infinitely many classes of finite size. Therefore, for infinitely many values of $n$, there will be some $m \in S_n$ where $B_\varepsilon$ succeeds. Furthermore, using that whenever $B_\varepsilon$ succeeds it does so except with an exponentially small failure probability, we have that with very high probability $m'(n)$ will be the smallest such $m$ on each interval $S_n$ where $B_\varepsilon$ succeeds somewhere. Consequently, $A_\varepsilon$ either fails with probability $1$ on a bad interval $S_n$ (where $B_\varepsilon$ fails everywhere), or it outputs with high probability a fixed sample from $\mathcal{D}_n$. The last conclusion relies also on the definition of $A_\varepsilon$ and on the fact that $B_\varepsilon$ is a (infinitely often) zero-error pseudodeterministic constructor for $Q$. (We observe that $A_\varepsilon$ is not zero-error because $B_\varepsilon$ may succeed on more than one input length in $S_n$, and it can happen with negligible probability that $A_\varepsilon$ will use a different substring of length $n^c$ during its computation.)

Finally, by taking $\varepsilon > 0$ arbitrarily small and using that $c$ is constant, it follows that $\mathfrak{D}$ admits subexponential time randomized algorithms that output a canonical sample from $\mathcal{D}_n$ for infinitely many values of $n$.
\end{proof}

We describe here an immediate application of Theorem \ref{t:canonical_sampler} that might be of independent interest in complexity theory, in certain situations where one has to select a hard function among a family of efficiently representable functions.

Suppose there is a polynomial time randomized algorithm that on input $1^n$ outputs some Boolean circuit with non-negligible probability, and fails otherwise (the algorithm is allowed to output different circuits among different executions). Let $C_n$ be the random variable that denotes either ``$\esymb$'' or the polynomial size circuit output by this algorithm, and $f_{C_n}$ be the corresponding Boolean function whenever $C_n \neq \esymb$.

It follows from Theorem \ref{t:canonical_sampler} that there is a subexponential time randomized algorithm that selects a canonical circuit $D_n$ for infinitely many values of $n$, and that outputs a circuit computing, say, the constant $0$ function on the remaining values of $n$. In particular, if every circuit $E_n$ in the support of the original algorithm computes a Boolean function $f_{E_n}$ with a certain desired property, there is a \emph{fixed} Boolean function $h \colon \{0,1\}^* \to \{0,1\}$ in $\mathsf{BPSUBEXP}$ that has the same property infinitely often. (Observe though that if on every $n$ only most functions $f_{C_n}$ have some property,  but not all of them, we cannot guarantee that $h$ will share this property on infinitely many input lengths.)

\subsection{Unconditional Constructions that are Explicit}
\label{s:explicit}

In this section, we show how to make our constructions for easy dense properties explicit, in the sense that the algorithm implementing the construction is an explicit algorithm. Consider for instance the problem of pseudodeterministically generating primes. Since Theorem \ref{t:primes} establishes that some algorithm runs in sub-exponential time and outputs a canonical prime infinitely often, a natural approach would be to employ a universal search procedure that runs all algorithms with short descriptions until a prime is produced. Unfortunately, this idea does not seem to work when the algorithms involved are randomized and we would like to maintain pseudodeterminism. 

We will employ a different strategy which will actually give us a bit more. In addition to addressing the issue of explicitness, we also control the gaps between input lengths on which the construction succeeds. However, this comes at the cost of tailoring the construction to a specific easy dense property, and the proof becomes more intricate. For simplicity, we will focus on bounded-error pseudodeterministic constructions for $\mathsf{Primes}$. This corresponds to a simpler version of Theorem \ref{t:uncond_hittingset}, where we not consider the hitting set family $\mathcal{H}^{\mathsf{easy}}$ obtained using the easy witness method, and do a win-win analysis based on the hitting set family $\mathcal{H}^{\mathsf{hard}}$ rather than a win-win-win analysis. 
In this variant setting, we will consider \emph{bounded-error} pseudodeterministic polynomial time constructions for $\mathsf{Primes}$, rather than \emph{zero-error} ones.

Recall that a fundamental issue with obtaining an explicit algorithm using the proof of Theorem \ref{t:generic} is that we do not know which of the worlds {\bf PSEUDO} and {\bf SPARSE} we live in (Section \ref{s:discussion}). There is an explicit algorithm corresponding to the world {\bf SPARSE}, but we only obtain an explicit algorithm corresponding to the world {\bf PSEUDO} if the algorithm for {\bf SPARSE} fails on all large enough input lengths, and we do not know a priori if this is the case.

Imagine the following ideal situation: the win-win analysis we carry out works input length by input length. Namely, for each large enough input length $n$, a given candidate hitting set $H_n$ constructible in deterministic subexponential time works, or else a different candidate hitting set $H'_n$ constructible in pseudodeterministic subexponential time works. If we were in this ideal world, we would get an explicit construction for each large enough length as follows. We first test each element in $H_n$ for primality, in some fixed order. If at least one of the tests succeed, we output the first element satisfying a test. If not, we generate $H'_n$ and again test the elements in some fixed order for primality. Now we are guaranteed to succeed by the assumption that the win-win analysis succeeds on each large enough input length, and as $H'_n$ is generated pseudodeterministically, we will output a fixed prime with high probability.

However, we are quite far from being in this ideal situation. Indeed, our argument that a pseudodeterministic algorithm succeeds relies on the hitting set family failing for all large enough input lengths, rather than on a single input length $n$. This enables us to obtain the complexity collapse $\mathsf{PSPACE} = \mathsf{BPP}$ and apply Lemma \ref{l:cond_discset}.

If we are to have any hope of controlling the set of input lengths on which the construction succeeds using such an argument, we need to mitigate this issue. Note that if we are only interested in a pseudodeterministic construction in subexponential time, the collapse $\mathsf{PSPACE} = \mathsf{BPP}$ is overkill; it is enough to have $\mathsf{PSPACE} \subseteq \mathsf{BPSUBEXP}$.

Consider the $\mathsf{PSPACE}$-complete language $L^{\star}$ in the statement of Theorem \ref{t:uniform_hard_rand}.
The first element of our new argument is a refined version of Theorem \ref{t:uniform_hard_rand}, which for any $\delta > 0$, yields a probabilistic algorithm solving $L^{\star}$ correctly on inputs of length $n$ in time $2^{n^{\delta}}$ {\it assuming} that the hitting set family $\{H_{\ell}\}$ fails at all input lengths $\ell \in [n^{1/D}, n^D]$, where $D$ is some constant depending on $\delta$. Thus we now only need the failure of the hitting set family on some polynomially bounded range of input lengths to obtain a complexity collapse consequence, albeit a milder one than before.

We also observe that this refined version can be used in an alternative argument for generating primes pseudodeterministically, by reducing the search version of $\mathsf{Primes}$ on input length $n$ to the $\mathsf{PSPACE}$-complete language $L^{\star}$ on some polynomially larger input length $n^k$. Hence, if we knew that the probabilistic algorithm based on the failure of the hitting set family for a polynomially bounded range of input lengths solved $L^{\star}$ correctly at some fixed input length $n^k$, we would be able to construct primes pseudodeterministically at length $n$ in subexponential time.

However, we have no easy way of knowing this. The straightforward method would be to explicitly test the success of the hitting set family on the appropriate range of input lengths, but this could take more than exponential time. 

Imagine instead the pseudodeterministic algorithm we wish to define being given a single {\it advice bit} per input length. If this advice bit is $0$ at length $n$, it indicates to the algorithm $A(1^n)$ that the hitting set family does indeed fail on all input lengths in $[n^{k/D}, n^{kD}]$; if the advice bit is $1$, it indicates that the hitting set family succeeds somewhere on that range. The point is that the requisite information is just a single bit depending on the input length $n$. The advice bit can be thought of as information for the algorithm about whether the world looks {\it locally} like {\bf PSEUDO} or {\bf SPARSE}, even if we do not know what the global picture is.

If the algorithm somehow had access to this advice bit, it could act as follows: if the advice bit were $0$, it would know that the probabilistic algorithm given by the refined version of Theorem \ref{t:uniform_hard_rand} solves $L^{\star}$ correctly at input length $n^k$, and by using the reduction from the search version of $\mathsf{Primes}$ to $L^{\star}$ and simulating the probabilistic algorithm when needed, it could pseudodeterministically output a prime in subexponential time. If the advice bit were $1$, ``all bets are off'', and the algorithm simply halts without an output.

For those readers familiar with the work on hierarchies for probabilistic polynomial time with advice \citep{Barak02, FortnowSanthanam04}, the use of a advice bit here might be reminiscent of that work. The similarity is that the advice bit is a way around constructibility issues, but the details are different.

An advantage in our setting is that while the advice bit might be conceptually useful, it is {\it not really needed}. The reason is that while the algorithm might not have the time to check if the hitting set family fails on all input lengths in a polynomially large range around $n$, it certainly can check if $H_n$ is a hitting set for $\mathsf{Primes}_n$ in deterministic subexponential time. If it is, the algorithm outputs the first prime in $H_n$, and we are done. If not, then the algorithm behaves {\it as if} the advice bit were $0$. The algorithm with this behaviour will not always be correct, but it will always succeed on some input length in any polynomially large enough interval of input lengths. Moreover, the algorithm is \emph{explicit}. We are exploiting here the fact that the world {\bf SPARSE} is a deterministic world, and that we can check \emph{deterministically} and not too inefficiently whether a given hitting set works at an input length.

We now give details, but first some extra notation. 

\vspace{0.2cm}
\noindent \textbf{Polynomial Gaps.} We call a set $S \subseteq \mathbb{N}$ {\it polynomially gapped} if $S$ is non-empty and there is a constant $k > 1$ such that for any $n \in S$, there is $m \in S$, $n < m \leq n^k$. 
\vspace{0.2cm}

We require the following refinement of Theorem \ref{t:uniform_hard_rand}, which holds for the same language $L^\star$ discussed before.

\begin{theorem}
 \label{t:uniform_hard_quant}
For any integers $b, c \geq 1$, there exists an integer $d \geq 1$ and a function $G^\star \colon \{0,1\}^* \to \{0,1\}^*$ with restrictions 
$$G_\ell^\star \colon \{0,1\}^\ell \to \{0,1\}^{m(\ell)},\quad \text{where}\;\;\;m(\ell) = \ell^b,$$ such that $G^\star$ can be computed in time $O(m(\ell)^d) = \mathsf{poly}(\ell)$ when given oracle access to $L^\star_{\leq \ell}$, and the following holds.  For every $\delta > 0$, there is a $\delta' > 0$ and a probabilistic algorithm $B_{\delta}$ such that for any large enough $n \in \mathbb{N}$ for which the output of $G^\star_\ell$ can be $(1/m(\ell)^c)$-distinguished from random for every $\ell \in [n^{\delta'},n^3]$ by an algorithm $A$ running in time $O(m(\ell)^c)$, $B_{\delta}$ when given access to $1^n$, $x \in \{0,1\}^{\leq n}$ and to the description of $A$, runs in time $O(2^{n^{\delta}})$, and computes $L^{\star}(x)$ with error at most $1/n^2$ over its internal randomness.
\end{theorem}

\begin{proof}
We give only a sketch, as the proof refines the proof of Theorem \ref{t:uniform_hard_rand}. The proof of Theorem \ref{t:uniform_hard_rand} proceeds by showing that a distinguisher for the output of $G^\star_\ell$ can be used to learn circuits for $L^{\star}$ on input length $n(\ell)$ polynomially related to $\ell$. By using the random self-reducibility and downward self-reducibility properties of $L^{\star}$ and hardness amplification, a distinguisher implies a polynomial-time probabilistic oracle algorithm that outputs circuits for $L^{\star}_{n(\ell)}$, where the oracle algorithm only makes $L^{*}$-queries of length $< n(\ell)$. By using the distinguishing property for each length $r \in [1, \ell]$, circuits for $L^{\star}$ can be learned iteratively in probabilistic polynomial time for input lengths from $1$ to $n(\ell)$, and thus $L^{\star}$ can be decided in probabilistic polynomial time on any input of length $n(\ell)$.

Suppose that we wish to compute $L^{\star}$ on inputs of length $|x| \leq n$. The main idea in our refinement here is to begin the iteration at input length $n^{\delta'}$, where $\delta'$ is chosen depending on $\delta$, so that a circuit of size $2^{n^{\delta'}}$ for $L^{\star}$ at length $n^{\delta'}$ can be computed in time $2^{n^{O(\delta)}}$ using brute-force search and the fact that $L^{\star}$ is in polynomial space. Now we use the distinguishing property for each length $r \in [n^{\delta'}, n^3]$ to obtain learners for corresponding input lengths for $L$, and thus iteratively build circuits for $L^{\star}$ for all input lengths up to $n$. Then it is enough to run the circuit for length $|x|$ to evaluate $L^{\star}$ on any input of that length. The total time taken is $O(2^{n^{\delta}})$, if we choose $\delta'$ sufficiently small as a function of $\delta$.
\end{proof}

\begin{theorem}[Restatement of Theorem \ref{t:explicit}]
\label{t:uncond_explicit}
For every $\varepsilon > 0$, there is a polynomially gapped set $S$ and a \emph{(}bounded-error\emph{)} pseudodeterministic construction for $\mathsf{Primes}$ on $S$, running in time $O(2^{n^{\varepsilon}})$.
\end{theorem}

\begin{proof}
Let $\varepsilon > 0$ be any constant. We show that there is a polynomially gapped set $S$ and a pseudodeterministic algorithm $A_{\varepsilon}$ on $S$ such that for each $n \in S$, the canonical output of $A_{\varepsilon}(1^n)$ is an $n$-bit prime, and moreover $A_{\varepsilon}$ always halts in time $O(2^{n^{\varepsilon}})$.

Define the language $\mathsf{LexFirstPrime}$ to consist of all tuples $<1^n,i>$ such that the $i$'th bit of the lexicographical first $n$-bit prime is $1$, where $n \geq 2$. By Bertrand's Postulate, $\mathsf{LexFirstPrime}$ is well-defined. It is easy to see that this language is decidable in polynomial space, as follows. Enumerate the $n$-bit integers in order and check each one for primality until an integer $p_n$ is found that passes the primality test. Accept on input $<1^n,i>$ iff the $i$'th bit of $p_n$ is $1$. Since $\mathsf{LexFirstPrime}$ is in $\mathsf{PSPACE}$ and $L^{\star}$ is $\mathsf{PSPACE}$-complete, there is a constant $k \geq 1$ such that $\mathsf{LexFirstPrime}$ reduces to $L^{\star}$ in deterministic time $n^k$.

Let $C > 0$ be an integer to be determined later. We partition $\mathbb{N}$ into intervals $I_i$, where $I_i = (2^{C^{i-1}}, 2^{C^i}]$ for $i \geq 1$, and $I_0 = [1,2]$. We define the algorithm $A_{\varepsilon}$ and show that it satisfies the required properties for at least one input length in each $I_i$, when $i$ is large enough. The algorithm operates in two phases, the first of which is deterministic and the second probabilistic.

Let $a > 0$ be a constant such that $L^{\star}$ is computable in deterministic time $2^{\ell^a}$ on inputs of length $\ell$, and let $c > 1$ be a constant such that the Primality algorithm of \cite{Agrawal02primesis} runs in deterministic time $n^c$. $A_{\varepsilon}$ operates as follows on input $1^n$. It first invokes the generator $G^{\star}$ from Theorem \ref{t:uniform_hard_quant} using parameters $b = \lceil 2a/\varepsilon \rceil$ and $c$ as chosen above, on input length $\ell = \lceil n^{\varepsilon/2a}\rceil$. It computes $H_n = \{ \mathsf{left}_n(u) \mid u \in G^{\star}(\{0,1\}^{\ell})\}$ in time $O(2^{n^{\varepsilon}})$, exploiting the efficiency guarantee for $G^{\star}$ from Theorem \ref{t:uniform_hard_quant} and the fact that $L^{\star}$ is computable in deterministic time $2^{\ell^a}$. It checks each element of $H_n$ in lexicographic order for primality, outputting the first $n$-bit prime in $H_n$ that it finds, if such a prime exists. Given $H_n$, the total time required for this testing is $O(2^{\ell} \cdot \mathsf{poly}(n))$, which is $O(2^{n^{\varepsilon}})$.

If no element of $H_n$ is prime, $A_{\varepsilon}$ commences its probabilistic phase. It sets $\delta = \varepsilon/k$ in the second part of Theorem \ref{t:uniform_hard_quant}; let $\delta' < \delta$ be the corresponding constant given by the theorem. $A_{\varepsilon}$ attempts to compute an $n$-bit prime in probabilistic time $O(2^{n^{\varepsilon}})$ as follows. It tries to determine for each $i$ satisfying $1 \leq i \leq n$, whether $<1^n,i>\;\in \mathsf{LexFirstPrime}$ by using the reduction from $\mathsf{LexFirstPrime}$ to $L^{\star}$, which produces instances of length $\leq n^k$. It answers each query $x$ to $L^{\star}$ by assuming that $\mathsf{Primes}$ $1/\ell^{bc}$-distinguishes $G^{\star}_{\ell}$ from random for each $\ell \in [n^{k \delta'}, n^{3k}]$, where $b$ is as defined in the previous para, and running the corresponding algorithm $B_{\delta}$ on $1^{n^k}$, $x$, and the code of the AKS primality algorithm. By Theorem \ref{t:uniform_hard_quant}, if $\mathsf{Primes}$ does indeed distinguish the output of the generator from random for the given range of input lengths, the algorithm $B_{\delta}$ decides $L^{\star}(x)$ correctly with error at most $1/n^{2k}$, since $x \in \{0,1\}^{\leq n^{k}}$. Hence, in this case, by a simple union bound, all $n$ queries of $A_{\varepsilon}$ to $B_{\delta}$ are answered correctly with probability at least $1-1/n$, using the fact that $k \geq 1$, and hence $A_{\varepsilon}$ correctly determines all the bits of the lexicographically first $n$-bit prime $p_n$ with error at most $1/n$. Thus, in this case, a fixed prime $p_n$ is output with probability at least $1-1/n$, which fulfils the bounded-error pseudodeterministic guarantee for $A_{\varepsilon}$. Using that $\delta = \varepsilon / k$ and the bound on the running time of $B_{\delta}$ given by Theorem \ref{t:uniform_hard_quant}, it follows that $A_{\varepsilon}$ halts in time $O(2^{n^{\varepsilon}})$.

We argue that for each interval $I_i$ of input lengths for $i$ large enough, there is $n_i \in I_i$ such that either one of the elements of $H_{n_i}$ is prime, or $A_{\varepsilon}$ outputs a fixed prime with high probability using the reduction to $L^{\star}$ as in the previous para. Note that in the first case, the deterministic phase of the algorithm has an output and the algorithm does not enter its probabilistic phase, while in the second case, the probabilistic phase has a fixed output with high probability. In either case, $A_{\varepsilon}$ operates pseudodeterministically on input $1^{n_i}$ and outputs a prime.

We set $C$ to be $\lceil 3k/\delta' \rceil$, where $k$ and $\delta'$ are as above. If there is no $n_i$ in $I_i$ such that at least one of the elements of $H_{n_i}$ is prime, and if $i$ is large enough, then it is indeed the case that $\mathsf{Primes}$ $1/\ell^{bc}$-distinguishes $G^{\star}_{\ell}$ from random for each $\ell \in [n_i^{k \delta'}, n_i^{3k}]$, where $n_i = 2^{C^i/3k}$, just using the fact that $\mathsf{Primes}$ is $1/n^c$-dense for large enough $n$. Hence, in this case, $A_{\varepsilon}$ does output the lexicographically first prime on $n_i$ bits with probability $1-o(1)$, which concludes the argument.
\end{proof}

To compare Theorem \ref{t:uncond_explicit} to Corollary \ref{c:uncond_primes}, the advantages of the former are that the algorithm is explicit, and that the input lengths for which it is guaranteed to produce primes are not too far apart. However, a somewhat subtle advantage of Corollary \ref{c:uncond_cons} is that the construction is guaranteed never to output two different primes on any input length -- it either outputs a fixed prime with high probability, or does not output a prime at all. With the construction of Theorem \ref{t:uncond_explicit}, this might not be the case. The algorithm has the bounded-error pseudodeterministic guarantee on at least one input length in each large enough interval, but there is no guarantee on the behaviour of the algorithm for other input lengths in the interval. This situation can be improved using Proposition \ref{p:purifying}.

Finally, we remark that one of the difficulties involved in proving a zero-error version of Theorem \ref{t:uncond_explicit} is that it seems one would need to convert a \emph{sub-exponential time} bounded-error randomized computation into a zero-error computation, as opposed to the relevant simulation behind the proofs of Theorem \ref{t:uncond_hittingset} and Corollary \ref{c:uncond_primes}, which is concerned with polynomial-time computations. 

\section{Pseudodeterminism and Derandomization}\label{s:pseudo_derand}

In order to state the results of this section we will need a few additional definitions. Here we work with ensembles $\mathfrak{D} = \{\mathcal{D}_n\}$ of distributions, where  we assume that each $\mathcal{D}_n$ is supported over $\{0,1\}^n$. Moreover, we say that such a sequence of distributions is polynomial-time samplable if there is a randomized polynomial-time algorithm $B$ (the sampler) such that for each $n \in \mathbb{N}$ and each $y \in \{0,1\}^n$, $\Pr[B(1^n) = y] = \mathcal{D}_n(y)$, where $\mathcal{D}_n(y) \eqdef \Pr[y \in \mathcal{D}_n]$. As usual, we use $U_n$ to refer to the uniform distribution on $n$-bit strings, which is clearly polynomial-time samplable. In some cases we view elements of $\{0,1\}^n$ as descriptions of Boolean circuits of size at most $n$, under some natural encoding. We may informally refer to $\mathfrak{D}$ as a distribution instead of as an ensemble of distributions.

We define various notions of derandomization on average over polynomial-time samplable distributions $\mathfrak{D}$. Our setting closely mirrors that of Impagliazzo-Wigderson \citep{DBLP:journals/jcss/ImpagliazzoW01}, and our proofs are inspired by their ideas.\\

\noindent \textbf{Average-Case Definitions.} Let $\ell\colon \mathbb{N} \rightarrow \mathbb{N}$ be a function. We say that a sequence $\mathfrak{G} = \{G_n\}$, where each $G_n\colon \{0,1\}^{\ell(n)} \rightarrow \{0,1\}^n$, is a PRG (resp.~i.o.PRG) on average over a distribution $\mathfrak{D}$ of Boolean circuits if for each $c > 0$ and for large enough $n$ (resp.~for infinitely many $n$), $G_n$ $(1/10)$-fools $C_n$ with probability at least $1-1/n^c$ over $C_n \sim \mathcal{D}_n$. We call $\ell(n)$ the seed length of the PRG.

Let $T\colon\mathbb{N} \rightarrow \mathbb{N}$ be a time bound. We say that the Circuit Acceptance Probability Problem ($\mathsf{CAPP}$) \emph{is solvable in time} $T$ (resp.~solvable infinitely often in time $T$) on average over $\mathfrak{D}$ if for all $c > 0$ there is a \emph{deterministic} algorithm $A$ running in time $T(n)$ such that for all $n \in \mathbb{N}$ (resp.~for infinitely many $n$), $ \Pr_{C_n \sim \mathcal{D}_n} [|A(C_n) - \Pr_{x \sim U_n}[C_n(x)=1]| < 1/10] \geq 1-1/n^c$. 

Let $L \subseteq \{0,1\}^{*}$ a language. We say that $L$ \emph{is solvable in time} $T$ (resp.~solvable infinitely often in time $T$) on average over $\mathfrak{D}$ if for all $c > 0$ there is a \emph{deterministic} algorithm running in time $T(n)$ which for all $n$ (resp.~infinitely many $n$) solves $L_n$ with success probability at least $1-1/n^c$ over $D_n$. Given a function $a\colon\mathbb{N} \rightarrow \mathbb{N}$, we also use the notion of being solvable in time $T$ with $a(n)$ bits of advice on average over $\mathfrak{D}$ -- here the algorithm solving $L$ gets access to an auxiliary advice string of length $a(n)$ which depends only on the input length.

We say that $\mathsf{CAPP}$ \emph{is solvable pseudodeterministically in time} $T$ (resp.~solvable infinitely often pseudodeterministically in time $T$) on average over $\mathfrak{D}$ if for all $c > 0$ there is a \emph{randomized} algorithm $A$ running in time $T(n)$ such that for all $n$ (resp.~infinitely many $n$), with probability at least $1-1/n^c$ over $C_n \sim D_n$, $A(C_n)$ outputs the same number $f(C_n)$ with probability $1-o(1)$ over its internal randomness, and $f(C_n)$ is a $(1/10)$-additive approximation to the acceptance probability of $C_n$.\\

The lemma below is implicit in \citep{DBLP:journals/cc/TrevisanV07}, which itself uses a variation of the argument in \citep{DBLP:journals/jcss/ImpagliazzoW01}. We omit the proof because it is almost identical to the proof of Theorem \ref{t:uniform_hard_rand}.

\begin{lemma} \emph{\citep{DBLP:journals/jcss/ImpagliazzoW01, DBLP:journals/cc/TrevisanV07}}
\label{l:PRGUnif}
For each $\varepsilon > 0$ there is a sequence $\mathfrak{G} = \{G_n\}$, where $G_n\colon \{0,1\}^{n^{\varepsilon}} \rightarrow \{0,1\}^n$ and $\mathfrak{G}$ is computable in time $2^{O(n^{\varepsilon})}$, such that if there is a polynomial-time samplable distribution $\mathfrak{D} = \{\mathcal{D}_n\}$ of Boolean circuits and a constant $c > 0$ for which for all large enough $n$, with probability $\geq 1/n^c$ over $C_n \sim \mathcal{D}_n$, $C_n$ is a $(1/10)$-distinguisher for $G_n$, then $\mathsf{PSPACE} = \mathsf{BPP}$.
\end{lemma}

We prove the following unconditional result on the (randomized) pseudodeterministic complexity of $\mathsf{CAPP}$.

\begin{theorem}[Restatement of Theorem \ref{t:capp}]
\label{t:pseudoCAPP}
For each $\varepsilon > 0$ and polynomial-time samplable distribution $\mathfrak{D}$, $\mathsf{CAPP}$ is solvable infinitely often pseudodeterministically in time $2^{O(n^{\varepsilon})}$ on average over $\mathfrak{D}$.
\end{theorem}

\begin{proof}
Let $\mathfrak{D}$ be any polynomial-time samplable distribution and $\varepsilon > 0$ be any constant. We show that at least one of the following holds: (1) $\mathsf{CAPP}$ is solvable pseudodeterministically in time $\poly(n)$ in the worst case, or (2) $\mathsf{CAPP}$ is solvable infinitely often deterministically in time $2^{O(n^{\varepsilon})}$ on average over $\mathfrak{D}$.

Let $\mathfrak{G} = \{G_n\}$ be the sequence of generators  given by Lemma \ref{l:PRGUnif}, and consider the algorithm $A$ that works as follows given a input circuit $C$ represented by a string of length $n$. $A$ counts the fraction of outputs of $G_n$ accepted by $C$, and outputs this fraction. If $G_n$ is an i.o.PRG on average over $\mathfrak{D}$, the algorithm $A$ solves $\mathsf{CAPP}$ infinitely often on average over $\mathfrak{D}$, since $A$ does not output a $1/10$-approximation for $C$ if and only if $C$ $(1/10)$-distinguishes the output of $G_n$ from a random $n$-bit string. The algorithm $A$ can be implemented in time $2^{O(n^{\varepsilon})}$ using the fact that $G_n$ is computable in time $2^{O(n^{\varepsilon})}$. Thus item (2) holds in this case.

If $\mathfrak{G}$ is not an i.o.PRG on average over $\mathfrak{D}$, then we can apply Lemma \ref{l:PRGUnif} to get $\mathsf{PSPACE} = \mathsf{BPP}$. In this case, by a simple translation argument, we have that $\mathsf{ESPACE} = \mathsf{BPE}$, and since $\mathsf{ESPACE}$ requires circuits of size $\geq 2^{m/2}$ on all large enough $m$ by direct diagonalization, we have that $\mathsf{BPE}$ requires circuits of size $2^{\Omega(m)}$ on all large enough $m$. Now by Lemma \ref{l:cond_discset} and Remark \ref{r:nonuniform}, we have a pseudodeterministic polynomial time algorithm which outputs a discrepancy set family which works for any circuit $C$ of size $n$ for large enough $n$, and hence again by outputting the fraction of elements in the discrepancy set which belong to $C$, we get a worst-case pseudodeterministic polynomial time algorithm solving $\mathsf{CAPP}$.  
\end{proof}

Finally, we establish the following equivalences. Note that all algorithms mentioned below are \emph{deterministic}.

\begin{theorem}[Restatement of Theorem \ref{t:equiv}]
\label{t:deranequivalence}
The following statements are equivalent\emph{:}
\begin{enumerate}
\item[\emph{1.}] For each polynomial-time samplable distribution $\mathfrak{D}$ of Boolean circuits and each $\varepsilon > 0$, there
is an \emph{i.o.PRG} $\mathfrak{G}$ on average over $\mathfrak{D}$ with seed length $n^{\varepsilon}$ that is computable in time $2^{O(n^{\varepsilon})}$.

\item[\emph{2.}] For each polynomial-time samplable distribution $\mathfrak{D}$ over Boolean circuits and each $\varepsilon > 0$, $\mathsf{CAPP}$ is solvable infinitely often in time $2^{O(n^{\varepsilon})}$ on average over $\mathfrak{D}$.

\item[\emph{3.}] For each polynomial-time samplable distribution $\mathfrak{D}$ over input strings and each $\varepsilon > 0$, $\mathsf{BPP}$ is solvable infinitely often in time $2^{O(n^{\varepsilon})}$ with $O(\log(n))$ bits of advice on average over $\mathfrak{D}$.

\item[\emph{4.}] For each $\varepsilon > 0$, $\mathsf{BPP}$ is solvable infinitely often in time $2^{O(n^{\varepsilon})}$ on average over $U_n$.

\end{enumerate}

\end{theorem} 

\begin{proof} The equivalence is established by the following chain of implications.

\vspace{0.2cm}
$(1) \Rightarrow (2)$: Fix a distribution $\mathfrak{D}$, and let $\mathfrak{G}$ be some PRG for $\mathfrak{D}$ that is guaranteed to exist by item (1). The algorithm for $\mathsf{CAPP}$ on a circuit $C_n$ of size $n$ simply runs $G_n$ on all seeds of size $n^{\varepsilon}$ and counts the fraction of seeds for which $C_n$ accepts. It outputs this fraction. Clearly, the algorithm can be implemented in time $2^{O(n^{\varepsilon})}$, as $G_n$ is computable in that amount of time. The correctness of the algorithm follows immediately from the guarantee on the PRG $\mathfrak{G}$ given by (1).

\vspace{0.2cm}
$(2) \Rightarrow (3)$: Fix a distribution $\mathfrak{D}$, let $L \in \mathsf{BPTIME}(n^k)$, where $k \geq 1$ is a fixed constant, and let $M$ be a bounded-error probabilistic Turing machine solving $L$ in time $n^k$ with error $\leq 1/3$ for $n$ large enough. To solve $L$ infinitely often in time $2^{O(n^{\varepsilon})}$ on average over $\mathfrak{D}$, we invoke the algorithm $A = A_{\mathfrak{D}', \varepsilon'}$ given by item (2) for the distribution $\mathfrak{D'}$ specified below, and parameter $\varepsilon' = \varepsilon/2k$. 

To sample from the distribution $\mathfrak{D'}$ on $n'$-bit strings, we first determine $n = \lfloor (n')^{1/2k} \rfloor$. We sample an input $x$ of length $n$ from $\mathcal{D}_n$, then compute the randomized circuit $C^x_{n^{2k}}$ obtained by applying the standard translation of randomized algorithms to circuits to the computation of $M$ on $x$. We then pad $C^x_{n^{2k}}$ in a standard way to an equivalent circuit $C^x_{n'}$ of size $n'$. $C^x_{n'}$ has input of length $n$, but we can simply pad the input length to $n'$ using dummy input bits. Clearly $\mathfrak{D'}$ is polynomial-time samplable, and hence there is an algorithm $A$ as above solving $\mathsf{CAPP}$ infinitely often on average over $\mathfrak{D'}$ in time $2^{O((n')^{\varepsilon'})}$, which is at most $2^{O(n^{\varepsilon})}$. 

We show how to solve $L$ infinitely often in time $2^{O(n^{\varepsilon})}$ with $O(\log(n))$ bits of advice on average over $\mathfrak{D}$. We define a deterministic machine $N$ taking $O(\log(n))$ bits of advice as follows. On input $x$ of length $n$, $N$ uses its advice to determine a length $n'$ such that $n^{2k} \leq n' < (n+1)^{2k}$. $N$ finds a randomized circuit $C^x_{n^{2k}}$ corresponding to the computation of $M$ on $x$ by performing the standard translation, and then pads this circuit in the standard way to a circuit $C^x_{n'}$ on $n'$ input bits. It applies the algorithm $A$ to $C^x_{n'}$, accepting if and only if $A$ outputs a number greater than $1/2$.

We are given that $A$ solves $\mathsf{CAPP}$ infinitely often in subexponential time on average over $\mathfrak{D'}$, and we would like to conclude that $N$ solves $L$ infinitely often in subexponential time on average over $\mathfrak{D}$. Indeed, let $\{n_i\}$ be an infinite sequence of input lengths on which $A$ solves $\mathsf{CAPP}$ in subexponential time on average over $\mathfrak{D'}$. It is not hard to see that $\{\lfloor (n_i)^{1/2k} \rfloor\}$ is an infinite sequence of input lengths on which $N$ solves $L$ in subexponential time with logarithmic advice on average over $\mathfrak{D}$.

\vspace{0.2cm}
$(3) \Rightarrow (4)$: It follows immediately from item (3) that for each $\varepsilon > 0$, $\mathsf{BPP}$ is solvable infinitely often in time $2^{O(n^{\varepsilon})}$ with $O(\log(n))$ advice on average over $U_n$, simply because $U_n$ is polynomial-time samplable. Corollary 7 and Corollary 9 of \citep{DBLP:journals/jcss/ImpagliazzoW01} then imply that for each $\varepsilon > 0$, $\mathsf{BPP}$ is solvable infinitely often in time $2^{O(n^{\varepsilon})}$ on average over $U_n$.

\vspace{0.2cm}
$(4) \Rightarrow (1)$: Here we use Lemma \ref{l:PRGUnif}. By diagonalization (cf.~Theorem 6 in \citep{DBLP:journals/jcss/ImpagliazzoW01}), we have that (4) implies $\mathsf{EXP} \neq \mathsf{BPP}$. Now there are two cases: either $\mathsf{EXP}$ does not have polynomial-size circuits, or $\mathsf{PSPACE} \neq \mathsf{BPP}$. Indeed, if both were false, we would have that $\mathsf{EXP} = \mathsf{PSPACE}$ (by the Karp-Lipton theorem for $\mathsf{EXP}$) and that $\mathsf{PSPACE} = \mathsf{BPP}$, which would together imply $\mathsf{EXP} = \mathsf{BPP}$, contradicting our assumption.

In the first case, by the hardness-randomness tradeoff of \citep{DBLP:journals/cc/BabaiFNW93}, it follows that for each $\varepsilon > 0$ there is an i.o.PRG with seed length $n^{\varepsilon}$, which is computable in time $2^{O(n^{\varepsilon})}$. Note that this i.o.PRG works even in the worst case, without a distributional assumption on circuits it fools. In the second case, we use Lemma \ref{l:PRGUnif} to conclude that for each polynomial-time samplable sequence $\mathfrak{D}$ of distributions  and for each $\varepsilon > 0$, there is an i.o.PRG with seed length $n^{\varepsilon}$ on average against $\mathfrak{D}$, computable in time $2^{O(n^{\varepsilon})}$. Hence in either case (1) follows, concluding our proof.
\end{proof}

\section{Further Directions}

We propose some directions for further research:

\begin{enumerate}

\item Theorem \ref{t:primes} is proved using general complexity-theoretic considerations, using no information about the set $\mathsf{Primes}$ apart from its polynomial density and its  decidability by a polynomial-time algorithm. The primes have been intensively studied, and a lot is known about their structure. Can this structural information be leveraged to prove stronger results about generating primes? Perhaps the technique of using complexity-theoretic pseudorandomness applied here could be combined with earlier ideas for generating primes deterministically to show stronger results.

\item Can other kinds of probabilistic algorithms be made pseudodeterministic? In very recent work, we use our ideas to give such algorithms unconditionally for various approximate counting problems. 

\item Are black-box derandomization, white-box derandomization and $\mathsf{BPP} = \mathsf{P}$ equivalent in the standard setting? Here, by the standard setting, we mean that we are interested in worst-case polynomial-time simulations that work for all large enough input lengths. As a first step toward this goal, it would be interesting to get the equivalence for average-case simulations, where we are even prepared to relax the ``works almost everywhere'' condition on the simulation to ``works infinitely often''.

\end{enumerate}

\section{Acknowledgements}

 We acknowledge useful discussions with Scott Aaronson, Valentine Kabanets, Jan Kraj\'{\i}\v{c}ek, James Maynard, Toni Pitassi and Amnon Ta-Shma.  Part of this work was done while the first author was visiting the second, funded by the second author's ERC Consolidator Grant no.~615075.

\bibliographystyle{alpha}	
\bibliography{refs}	

\end{document}